\documentclass[a4paper,12pt]{article}
\usepackage{amssymb,amsmath,amsthm}
\usepackage{graphicx}

\usepackage{amsfonts}
\usepackage{latexsym}
\usepackage{color}
\usepackage[english]{babel}
\usepackage{hyperref}
\usepackage{hyperref}
\usepackage{verbatim} 
\usepackage{bm} 
\usepackage{xcolor}
\usepackage{graphicx} 
\usepackage{subcaption}
\usepackage[font=footnotesize,labelfont=bf]{caption}
\usepackage{mathrsfs}
\definecolor{darkgr}{rgb}{0.0, 0.62, 0.42}

\numberwithin{equation}{section}

\newcommand{\Span}{\mathrm{span}}

\newtheorem{lemma}{Lemma}[section]
\newtheorem{theorem}[lemma]{Theorem}

\newtheorem{corollary}[lemma]{Corollary}
\newtheorem{remark}[lemma]{Remark}

\newtheorem{definition}[lemma]{Definition}

\def\tM{{\tt m}}

\def\bv{{\bf v}}
\def\bu{{\bf u}}

\def\simm#1{\mathop{\sim}^{#1}}
\def\sigmav{\varsigma}

\newcommand{\csi}{\xi}

\newcommand{\nform}{N}
\newcommand{\NForm}{\cN}
\def\mass{\mathcal M \kern -2pt a \kern-2pt d}
\def\russi{Par08,PS10,PS12}

\def\be{{\bf e}}

\def\simb#1{S^{#1,\delta}}

\def\im{{\rm i}}

\def\cV{{\mathcal{V}}}
\def\cN{{\mathcal{N}}}

\def\cC{{\mathcal{C}}}

\def\cG{{\mathcal{G}}}
\def\cE{{\mathcal{E}}}
\def\cW{{\mathcal{W}}}
\def\cL{{\mathcal{L}}}

\def\cR{{\mathcal{R}}}

\def\cB{{\mathcal{B}}}
\def\ops#1{{OPS}^{#1,\delta}}

\def\scalg#1#2{\left(#1;#2\right)_g}
\def\scal#1#2{\left(#1;#2\right)}
\def\scalgs#1#2{\left(#1;#2\right)_{g*}}
\def\norg#1{\left\|#1\right\|_g}
\def\norgs#1{\left\|#1\right\|_{g*}}

\def\tc{{\tt c}}

\def\td{{\tt d}}
\def\tN{{\tt N}}
\newcommand{\cQ}{\mathcal{Q}}

\newcommand{\R}{\mathbb R}

\newcommand{\Z}{\mathbb Z}

\newcommand{\N}{\mathbb N}
\newcommand{\T}{\mathbb T}

\def\norma#1{\Vert#1\Vert}
\def\norm#1{\Vert#1\Vert}
\def\xik{\xi^{\kappa}}
\def\zetak{\zeta^{\kappa}}

\newcommand{\Mc}{M^{(c)}}

\def\inte#1{\lfloor#1\rfloor}
\def\frazp#1{\left\{#1\right\}}

\def\perpz#1{\tilde{#1}}

\newcommand{\Op}{{Op}^{W}}
\newcommand{\OPS}{{OPS}\,}

\definecolor{awesome}{rgb}{1.0, 0.13, 0.32}

\def\change#1{{\color{black}#1}}

\newcommand{\ep}{\epsilon}
\newcommand{\Cip}{C^{\prime}}
\newcommand{\Dip}{D^{\prime}}

\newcommand{\intg}[2]{\left({#1}, {#2}\right)_{g} }
\newcommand{\scala}[2]{\left({#1}, {#2}\right) }

\newcommand{\gaugemap}{U_{\beta}}

\newcommand{\ii}{i}

\def\ac{{\mathcal A}\kern-.7pt\ell\kern-.9pt\mathcal{S}}

\begin{document}


\title{Spectral
asymptotics of all the eigenvalues of Schr\"odinger operators on flat tori}

\date{}


\author{ Dario Bambusi\footnote{Dipartimento di Matematica, Universit\`a degli Studi di Milano, Via Saldini 50, I-20133
Milano. 
 \textit{Email: } \texttt{dario.bambusi@unimi.it}}, Beatrice Langella\footnote{International School for Advanced Studies (SISSA), via Bonomea 265, I-34136 Trieste.
 \textit{Email: } \texttt{beatrice.langella@sissa.it}}, Riccardo Montalto \footnote{Dipartimento di Matematica, Universit\`a degli Studi di Milano, Via Saldini 50, I-20133
Milano.
 \textit{Email: } \texttt{riccardo.montalto@unimi.it}}
}

\maketitle

\begin{abstract} We study Schr\"odinger operators with Floquet
boundary conditions on flat tori obtaining a spectral result giving an
asymptotic expansion of all the eigenvalues. The expansion is in
$\lambda^{-\delta}$ with $\delta\in(0,1)$ for most of the eigenvalues
$\lambda$ (stable eigenvalues), while it is a ``directional
expansion'' for the remaining eigenvalues (unstable eigenvalues). The
proof is based on a structure theorem which is a variant of the one
proved in \cite{PS10,PS12} and on a new iterative quasimode argument.
\end{abstract} \noindent


\medskip



\tableofcontents
\section{Introduction}\label{intro}

The spectrum of periodic Schr\"odinger operators has been extensively
studied in the last decades and it is essentially fully understood in
dimension one. In particular a full asymptotic expansion of the
eigenvalues $(\lambda_j)_{j \in \Z}$ in the parameter $1/|j|^2$ has
been given by Marchenko \cite{marchenko}. In higher dimension the
situation is considerably more complicated. Consider the Laplacian
with periodic boundary conditions on a general torus
$\T^d_\Gamma:=\R^d/\Gamma$, with $\Gamma$ a maximal dimensional
lattice. Its eigenvalues are given by $\{\|\xi\|^2\}_{\xi\in\Gamma^*}$
with $\Gamma^*$ the dual lattice\footnote{We recall that the dual
lattice is defined as the set of $\xi$'s s.t. $\xi\cdot \gamma\in 2\pi
\Z$ $\forall \gamma\in\Gamma$} to $\Gamma$.  For generic lattices the
differences between couples of eigenvalues accumulate at zero and this
makes difficult to use standard resolvent expansions in order to
obtain properties of the eigenvalues. 

A milestone of the higher dimensional theory is the result {of
\cite{FKT90}, \cite{Fri90} (see also \cite{W})} who proved that,
provided $\cV$ is a sufficiently smooth potential with zero average,
and $\Gamma$ a generic lattice, most of the eigenvalues {of the
Laplace operator $-\Delta$ are stable under the perturbation given by
the potential $\cV,$ in the sense that there are two eigenvalues
$\lambda_{\pm \xi}$ of
\begin{equation} \label{H} -\Delta+\cV(x)\ ,
\end{equation}
in}
the interval
$$ \Big[ \|\xi\|^2 - \frac{1}{\|\xi\|^{{2 \delta}}}\,,\, \|\xi\|^2 +
\frac{1}{\|\xi\|^{{2 \delta}}} \Big] $$ with $\delta \in (0, 1)$ a
parameter. However, it was shown in \cite{FKT2} (developing an
argument by \cite{ERT}), that there are also eigenvalues which behave
differently and are not stable.

The stable eigenvalues also admit a full asymptotic expansion in 
 {$\displaystyle{\lambda^{-\delta} \sim \|\xi\|^{-2\delta}}$}, which can be obtained as a byproduct of the works
\cite{Par08,PS10,PS12} (see also \cite{Kar97} and \cite{vel15} for some
partial previous results) and is explicitly given in \cite{noi}.
For the unstable eigenvalues such an asymptotic expansion is simply
false. Here we address the problem of understanding the kind
of asymptotic expansion valid for unstable eigenvalues.

To present our approach we first recall the method developed in
\cite{PS10} (see also \cite{PS12}). In the paper \cite{PS10} the
authors developed a technique to construct a unitary transformation
which conjugates the operator \eqref{H} to a new operator which is the
sum of a ``normal form operator" and a remainder. Such a technique can
be interpreted as a quantization of the classical normal form
algorithm usually employed to study the dynamics of the Hamiltonian
system $h(x, \xi) := \|\xi\|^2 + V(x),$ whose quantization is
\eqref{H}. On the classical side, it is well known that the dynamics
(and thus the normal form) of a Hamiltonian system is completely
different in the resonant and in the nonresonant regions: it turns out
that stable eigenvalues correspond to the nonresonant regions, while
the unstable eigenvalues correspond to resonant regions. In
particular, {in \cite{PS10} a precise definition of resonant\slash
  nonresonant regions was given and it was shown that the normal form
  operator is a block diagonal operator which is just a Fourier
  multiplier if one localizes it in the nonresonant region; then the blocks
  corresponding to resonant regions turn out to be finite dimensional,
  but their dimension is not bounded.  When adding the remainder, a
  quasimode argument can be used to get the asymptotics of the
  eigenvalues corresponding to the nonresonant region, but almost
  nothing is known on the eigenvalues of the other blocks.
  
Here we want to obtain precise asymptotics also of all the
eigenvalues corresponding to the resonant regions.}
{The} asymptotic
expansions we get are not in the parameter
$\left\|\xi\right\|^{-2\delta}$: instead, they are directional
asymptotics. To explain this point, label the eigenvalues of \eqref{H}
using the points
$\xi\in\Gamma^*$, then, roughly speaking, the result is the following:
consider a submodule $M$ of $\Gamma^*$ and assume that the
 {vector $\xi\in\Gamma^*$ is resonant with the vectors
of a basis of $M$, but with no other vectors in $\Gamma^*$, then
the corresponding eigenvalue $\lambda_{\xi}$ admits an asymptotic
expansion in the parameter $\|(\xi)_M\|^{-2\delta}$, the lower index
$M$ denoting orthogonal projection on $M$.}

\vskip5pt

 {The main point in order to get such an expansion consists in first
proving a structure theorem which is a variant of the block diagonal
decomposition of \cite{PS10,PS12}, but \change{global in the Hilbert space and} suitable for
iteration.}  {This is needed in order to further decompose the
  resonant blocks in sub-blocks which at the end of the procedure will
  be just isolated points or finite dimensional, but with \emph{uniformly bounded dimension.}} \change{In view of possible future applications to quasilinear problems, we prove our structure theorem in the case where $\cV$ is not a potential, but a pseudodifferential operator of order strictly less than 2.}

More precisely, (as in \cite{PS10, PS12})  {as a first step }we
conjugate the operator in \eqref{H} to $\widetilde H+\cR$ with $\cR$ a
smoothing pseudodifferential operator and $\widetilde H$ a block
diagonal operator. The blocks corresponding to the non resonant
zone are  just isolated points.  The main novelty of our structure
  theorem is that we prove that in the non trivial blocks $\widetilde H$ is still a periodic Schr\"odinger operator, but on a lower
dimensional torus: essentially it contains only the angles in the
resonant directions. { We point out that a similar, but
  less precise property, was proved in
  \cite{PS09} just for the 2-d case.} Then \change{in the case where $\cV$ is bounded,} since in each block one has the same
structure as that of \eqref{H}, one can apply again the normal form
procedure and iterate until one is left with trivial blocks and blocks
with uniformly bounded dimension. However, since the new operator only
depends on the resonant angles one gets that the new normal form is
only up to a remainder which is smoothing in the resonant
directions. This is the source of the directional decay.

One further difficulty is that, since there are infinitely many blocks,
one must have a {uniform} control of all the constants of the restricted
operators. We will achieve this goal by performing the whole
construction in an intrinsic way: we define the resonant regions, the
blocks, and the seminorms of the pseudodifferential operators in terms
of the natural metric of the torus.  {This allows a control of all the
constants of the restricted operators in terms of the constants of the
original operator.}

The final step of the proof consists in reconstructing the eigenvalues
of the original operator. This is obtained through an iterative
quasimode argument  {that, as far as we know, is new}. To explain
it  {consider the case $d=2$};  {in this case, when restricting
  to the blocks, $\widetilde H$ turns out to be }either a Fourier
multiplier or a 1 dimensional Schr\"odinger operator. So essentially
everything is known on the spectrum of each block operator. However in
order to establish a correspondence between the spectrum of
$\widetilde H$ and the spectrum of $\widetilde H+\cR$ one must have
some information on how the eigenvalues of $\widetilde H$ are
distributed on the real line and to know something on the
eigenfunctions. The information that we use on the eigenvalues is just
Weyl law, which allows {to partition of the spectrum in
  clusters}. Then to each cluster we apply a new quasimode argument
(which is a development of that used in \cite{BKP15})  {which
  allows to describe
  how the perturbation changes the eigenvalues}. Concerning the
eigenfunctions, the information that we extract is that their negative
Sobolev norms decay fast with $\|{\xi}\|^{-1}$  {(see Equation
  \eqref{atima_soba} for a precise statement)}.  We point out that it
could be interesting to extract more information on the
eigenfunctions. 

We point out that we think that our formulation of the Structure
theorem could be useful also for more applications, for example we
think that one could get a detailed description of the semiclassical
measures \cite{Macia} {or (following \cite{Roy}) a precise
semiclassical expansion in $\hbar$ of the eigenvalues. }


\bigskip

 {Finally, we emphasize that our motivation for this research comes from our
work in KAM theory for PDEs: the construction of quasiperiodic
solutions of a Hamiltonian PDE requires a full understanding of the
dynamics of the operator obtained by linearizing the PDE at any
approximate solution. A good model problem is the time dependent
Schr\"odinger equation $ -\ii \partial_t u = -\Delta u+\cV(t,x) u$,
where $\cV$ is a smooth potential depending in a quasiperiodic way on
time and an efficient way to completely characterize its dynamics
consists in conjugating such an operator to a time independent equation
(reducibility problem). This can be done using a general strategy
developed in \cite{BBM14, 2, 1, BM, BBHM} for the study of quasilinear
1-d problems and extended to some very particular higher dimensional cases in
\cite{BGMR1,FGMP,BLM19,M,BGMR19,GrebFeola}. The first step of this
approach requires a very precise knowledge of the eigenvalues of the
problem in which time is frozen, and that's why we attack here this
problem. The final aim of this line of research is to bypass the limitation of the results of
\cite{bourglemma, BouLib,EK09,PX,BMas} and to get a KAM theory applicable to
equations on manifolds or domains which are as general as possible. }

\vskip 10pt

\noindent
{The paper is split in two parts: Part I, containing Sections \ref{spectral.res} and
\ref{spec sec}, in which we give our main results, and Part II
containing the proofs.}  In Section \ref{spectral.res}, we give a
statement of the Structure Theorem, recalling also the main notions
needed to give a precise statement.  In Section \ref{spec sec}, we start
by describing in detail the partition of $L^2(\T^d_\Gamma)$ in
invariant subspaces. This is the quantum analogue of the construction
of the geometric part of Nekhoroshev theorem. In particular this
is needed in order to give a precise statement of our spectral result
(see Theorem \ref{speci_spettri}): the kind of asymptotics that we
give depends on the block to which the eigenvalue belongs (in a sense
that will be made precise).

Part II is devoted to the proof of the main results. In
Sect. \ref{sez forma normale nonris} we give our normal form lemma
conjugating up to a smoothing operator \eqref{H} to a normal form
operator. This corresponds to the analytic part of Nekhoroshev's
theorem. In Section \ref{sez blocchi} we study the partition of
Subsect. \ref{geometrachilegge} in order to show that it is actually a
partition and is left invariant by an operator in normal form. This
corresponds to the geometric part of Nekhoroshev's theorem. Finally
in Sect. \ref{quasiquasi} we give a quasimode argument adapted to our
situation and prove our spectral result.

The paper contains also three appendixes: in Appendix \ref{pdc} we
adapt some standard results on pseudodifferential calculus to our
context, in Appendix \ref{lemmacci} we prove some very technical
lemmas which are used in the core of the paper, finally in Appendix
\ref{aquasi} we prove a couple of results on spectral problems needed
in Section. \ref{quasiquasi}.

\bigskip

{\sc Acknowledgments.} While working on this project we had many
discussions with several people. In particular we had several
enlightening discussions with Antonio Giorgilli, who explained to us
the details of the geometric part of the proof of Nekhoroshev's
theorem: its understanding has been a key step for the completion of
this work. A particular thank also goes to Thomas Kappeler who
introduced us to the literature on higher dimensional periodic
Schr\"odinger operators. We also thank Emanuele Haus, Fabricio Macia,
Alberto Maspero, Michela Procesi and
Didier Robert for some very stimulating discussions.

This work is partially supported by GNFM.
\vskip25pt

\Large  \noindent{\bf PART I: Statements}

\normalsize

\section{The structure theorem}\label{spectral.res}
\subsection{Preliminaries}

Let $\Gamma$ be a 
lattice of dimension $d$ in $\R^d$, with basis ${\bf e}_1, {\bf e}_2,
\ldots, {\bf e}_d$, namely 
\begin{equation}\label{definizione Gamma}
\Gamma := \Big\{  \sum_{i = 1}^d k_i {\bf e}_i : k_1, \ldots, k_d \in
\Z\Big\}\ ,
\end{equation}
and define 
\begin{equation}
  \label{toro}
 \T^d_\Gamma := \R^d / \Gamma\,.  
\end{equation}
Consider the Schr\"odinger operator
\begin{equation}
  \label{1.1}
-\Delta+\cV\ , 
\end{equation}
with Floquet boundary conditions on $\T^d_\Gamma$, namely acting on
functions $u$ which fulfill (together with their first derivatives)
the boundary conditions
\begin{equation}
  \label{floquet}
u(x+\gamma)=e^{\ii\gamma\cdot\kappa}u(x)\ ,\quad \forall \gamma\in
\Gamma ;
\end{equation}
$\kappa\in\T^d_{\Gamma^*}$ is a parameter.
Here $\cV$ is either a
potential, or more generally a pseudodifferential operator of order
zero on $T^*\T^d_\Gamma$ (see Definition \ref{def weyl berti} below for a precise definition).

By making the Gauge transformation $u= e^{i \kappa \cdot x}\tilde u$ the operator
\eqref{1.1} is conjugated to the operator
\begin{equation}
\label{per}
H=\sum_j(D_j+\kappa_j)^2+\cV\ ,\quad D_j:=-\im \partial_j 
\end{equation}
with periodic boundary conditions (p.b.c.) on $\T^d_{\Gamma}$; from
now on we will only use the variable $\tilde u$ and omit the tilde.  If $\cV\equiv
0$, then the eigenvalues of $H$ are simply given by
\begin{equation}
  \label{laple}
\lambda_\xi^{(0)}:=\norma{\xi+\kappa}^2\ ,\quad \xi\in \Gamma^*\ .
\end{equation}

By introducing in $\T^d_\Gamma$ the basis of the vectors
 $\be_i$, $H$ is reduced to 
\begin{align}
\label{g}
H=-\Delta_{g, \kappa}+\cV\ ,
\\
\label{deltag}
-\Delta_{g, \kappa}:=g^{AB}(D_A+\kappa_A)(D_B+\kappa_B)
\end{align}
with p.b.c. on the standard torus $\T^d:=\R^d/(2\pi \Z)^d$. Note that in formula \eqref{deltag}, we use the standard Einstein notation, namely 
$$
g^{AB}(D_A+\kappa_A)(D_B+\kappa_B) := \sum_{A, B = 1}^d g^{AB}(D_A+\kappa_A)(D_B+\kappa_B)
$$  
where 
\begin{equation}
\label{g.1}
g_{AB}:=\be_A\cdot\be_B, 
\end{equation}
and the matrix with upper indexes is the inverse of the matrix with
lower indexes, namely it is defined by
$$
\quad g_{AB}g^{BC}=\delta_A^C\ .
$$

Conversely, given an operator of the form \eqref{g}, by introducing a
basis which is orthonormal with respect to the metric $g:=(g_{AB})$
and making a Gauge transformation, one is reduced to a standard Schr\"odinger operator with Floquet boundary conditions on a suitable
torus $\T^d_{\Gamma}$. For this reason, from now we will restrict our
study to the operator \eqref{g} and we will call it a {\it
  Schr\"odinger operator of dimension $d$ with Floquet boundary
conditions}. Furthermore, with a slight abuse of language, we will use
the same name for operators which are the restriction of an operator
of the form \eqref{g} to a subspace of $L^2$.

 In the following we will only deal with
scalar products and norms with respect to the metric $g$. We will
denote 
\begin{equation}
  \label{scal}
\scalg xy:=g_{AB }x^Ay^B\ ,\quad \scalgs \xi \eta:= g^{AB}\xi_A\eta_B
  \end{equation}
the scalar product with respect
to this metric of two vectors $x,y$ or two covectors
$\xi,\eta$. Correspondingly we will denote  
\begin{equation}
  \label{nors}
\norg x^2:=\scalg xx\ ,\quad \norgs \xi^2:= \scalgs \xi\xi\ .
  \end{equation}
Finally we will denote by $d\mu_g(x)$ the volume form corresponding to
$g$.
The following constant plays a relevant role in our construction:
\begin{gather}
	\label{coerc}
	{\frak c} := \underset{ k \in \Z^d \backslash \{ 0 \}}{\inf
        }\norgs k^2 \,.
\end{gather}
{ Given $s$ linearly independent vectors $\{u_1, \dots, u_s\}$ in
  $\Z^d$, denote by $\textrm{ Vol}_{g^*} \{u_1 | \cdots | u_s\}$ the
  $s-$ dimensional volume, calculated with respect to the metric
  $g^*$, of the parallelepiped in $\R^d$ with edges given by $\{u_1,
  \dots, u_s\}$. A further relevant constant is
\begin{equation}\label{volmin}
{\frak C} := \min_{ 1 \leq s \leq d}\ \min_{u_1, \dots, u_s \in \Z^d} \textrm{ Vol}_{g^*} \{u_1 | \cdots| u_s\}\,.
\end{equation}
\begin{remark}\label{rmk cbig}
In Lemma \ref{lemma parallelepipedi} of the Appendix \ref{lemmacci},
we will prove that $\frak C$ is strictly positive.
\end{remark}
In the following we will often refer to the constants ${\frak c}$,
${\frak C}$ as the constants of the metric.}

\subsection{Pseudodifferential calculus}\label{pcal}
Given $u\in L^2(\T^d)$, we define as usual its Fourier series by
$$
u(x) = \sum_{\csi \in \Z^d}
\hat{u}_{\csi}\ e^{\ii \xi \cdot x}
$$
where $\xi \cdot x=\xi_Ax^A$ is the usual pairing
between a vector and a covector.\\
Fix $\kappa \in \R^d/\Z^d$, then we define $H^s(\T^d)$ to be the completion of
  $\cC^{\infty}(\T^d)$  in the norm
  \begin{equation}
\label{norfou}
\norma u ^2_{H^s} = \sum_{\xi\in\Z^d} {\langle{\xi {+ \kappa}}\rangle_g}^{2s}|\hat u_\xi|^2\ ,
\end{equation}
{where $\langle \xi \rangle_g := \left( 1 + \norgs\xi^2
	\right)^{1/2}$.}
Given a function $a\in
C^\infty(T^*\T^d)$, we define (exploiting the equivalence
$T^*\T^d\simeq\T^d\times\R^d $),
\begin{equation} \label{def norma}
\begin{aligned}
\| d^{M}_x d^{N}_\csi a(x, \csi)\| = \underset{\begin{subarray}{c}
	\norg{h^{(i)}}=1\\ \norgs{k^{(j)}} = 1 
	\end{subarray}} {\sup} \left|d^{M}_x d^{N}_\csi a(x, \csi) \left[h^{(1)}\,, \dots, h^{(M)}, k^{(1)}\,, \dots, k^{(N)}\right] \right|\,.
\end{aligned}
\end{equation}

\begin{definition} \label{def weyl berti} Let ${a \in {\cal C}^{\infty} \left( T^*\T^d\right)}$ and
	$m \in \R,\ \delta >0$ and $\kappa \in \R^d/\Z^d\,.$ We say that $a \in S^{m,\delta}$ is a
	symbol of order $m$, if $\forall\ N_1\,, N_2 \in \N\,,$ there exists
	a constant ${C_{N_1, N_2}>0}$ such that
	$$
	\| d_x^{N_1} d_{\csi}^{N_2} a(x, \csi)\| \leq C_{N_1, N_2} \langle \csi + \kappa \rangle_g^{m - \delta|N_2|} \quad \forall x\in \T^d\,,\ \csi \in \R^d\,.
	$$
	We also define $S^{-\infty,\delta}:=\cap_{m}S^{m,\delta}$. 
\end{definition}
\begin{remark}
	\label{kappain}
	The parameter $\kappa$ which appears in the definition of symbol and as a weight in the Sobolev norms \eqref{norfou} has
	been introduced in order to get uniform estimates suitable for
        the iteration of Theorem \ref{main}. 
\end{remark}

\begin{definition}\label{oppseudo} Let $a\in S^{m,\delta}$, its
	\emph{Weyl quantization} is the linear operator $A\equiv
	\Op(a)$ defined by
	\begin{equation} \label{def weyl toro}
	\left(\Op(a) [u]\right)(x) =  \sum_{\csi \in \Z^d} \sum_{h \in \Z^d} \hat{a}_h\left(\csi + \frac{h}{2}\right) \hat{u}_{\csi}\ e^{\ii (\csi + h) \cdot x}\,,
	\end{equation}
	where $\forall k \in \Z^d$ and $\forall \csi \in \R^d$
	$$
	\hat{a}_k(\csi) = \frac{1}{\mu_g(\T^d)} \int_{\T^d} a(x, \csi) e^{- \ii  k \cdot x} d\mu_g 
	\,.
	$$
\end{definition}

\begin{definition}\label{psd}
	Let $A$ be a linear operator on $L^2(\T^d)$, we say that it is a
	pseudodifferential operator of class $\ops m$ if there
	exists $a\in S^{m,\delta},$ such that $A = \Op(a)$. Operators
        of class $\ops{-\infty}$ will be called \emph{smoothing}.
\end{definition}

\begin{definition}[\bf Seminorms]\label{definizione seminorme}
	Let $a \in S^{m, \delta}$ and $N_1, N_2 \in \N$. We define 
	$$
	C_{N_1, N_2}(a) := \sup_{(x, \xi) \in \T^d \times \R^d} \langle \csi + \kappa \rangle_g^{ \delta N_2 - m} \|  d_x^{N_1} d_{\csi}^{N_2} a(x, \csi)\|\,.
	$$
	Equivalently, if $A = \Op(a),$ we set
	$
	C_{N_1, N_2}(A) := C_{N_1, N_2}(a)\,.
	$
\end{definition}
\begin{remark}\label{rmk semin}
$\{ C_{N_1, N_2}(\cdot) \}_{N_1, N_2 \in \N}$ is a family of seminorms
  on $S^{m, \delta}$, and we will refer to $\{C_{N_1, N_2}(A)\}_{N_1,
    N_2 \in \N}$ as the family of seminorms of the operator $A\,.$
  {All the definitions are given in such a way that the seminorms
    do not depend on the coordinates that one uses in $\T^d$, namely, if one
    changes the basis $\{ \be_i \}$ by means of a unimodular
    transformation $A$ (i.e. a unimodular matrix with integer coefficients), then this does not change the
    value of the seminorms. This is crucial for our procedure.}
\end{remark}
{We refer to the Appendix \ref{pdc} for some basic properties of
  pseudodifferential calculus {in the intrisic formulation}. In particular, we emphasize  that  all the
  constants controlling the seminorms of the composition, commutators, and
  exponentiation of pseudodifferential operators depend only on the
  constants of the metric. This is is needed for iterating the
  structure theorem.}

\subsection{Submodules, subspaces and statement of the Structure Theorem}\label{sub}

\begin{definition}
  \label{generated}
	Given $E \subseteq \Z^d\,,$ we denote  
	\begin{equation} \label{def cal e}
	{\cal E} = \overline{\Span{\{e^{\ii \csi \cdot x}\ |\ \csi \in E\}}}\,,
	\end{equation}
        where the bar denotes the closure in $L^2$. We will
       call such a subspace \emph{subspace generated by $E$. }
\end{definition}
\begin{definition}
  \label{gen.2}
        We will denote by $\Pi_{\cal E}:L^2(\T^d)\to \cE$ the
        orthogonal projector on ${\cal E}\,$ and, given a linear
        (pseudodifferential) operator $F$, we will write
\begin{equation}
  \label{restri}
F_{\cal E} := \Pi_{\cal E} F \Pi_{\cal E}\,.
\end{equation}
\end{definition}

The block decomposition as well as the spectral asymptotics of the
Schr\"odinger operator are related to the submodules of $\Z^d$, for
this reason we recall some properties of the bases of the
modules. {The systematic use of the properties of discrete
  submodules is one of the differences with the construction of
  \cite{\russi}. This plays a crucial role in order to show that the
  operator one obtains in each invariant block still has the structure
  of a Laplacian plus a potential plus a more regularizing
  pseudodifferential operator.}

\begin{definition}
  \label{mod}
A subgroup $M$ of $\Z^d$ is called a submodule if
${\Z^d\cap\Span_{\R}M=M.}$ Here and below, $\Span_{\R}M$ is the subspace generated
by taking linear combinations with real coefficients of elements of
$M$. 
\end{definition}
Given a discrete submodule $M$ of $\Z^d$ it is well known
that it admits a basis, namely that there exist $d'$ independent
vectors $\bv^1,...,\bv^{d'}$ such that
\begin{equation}\label{coordinate adattate bla}
M={\rm span}_{\Z}\{\bv^1,...,\bv^{d'} \} : = \Big\{w\in \Z^d\ :\ w= \sum_{k = 1}^{d'} n_k \bv^k, \quad n_1, \ldots, n_{d'} \in \Z \Big\}\ .
\end{equation}
\begin{definition}
  \label{integer}
  Given a basis {$\{\bv^k\}_{k =1, \ldots, d'}$ of $M$ and a vector
    $m=m_k\bv^k\in$span$_\R M$, we denote
  $$
\inte m:=\inte{m_k}\bv^k\ , 
$$
with $\inte{m_k}$ the integer part of $m_k$,} {and
	 $$
	\{ m\}:=\{m_k\}\bv^k\ , 
	$$
where $\{m_k\}$ is the fractional part of $m_k$.}
\end{definition}
{Given a covector $\csi \in \Z^d$, a Floquet parameter $\kappa,$ and a module $M$, we will have to
	decompose the covector $w = \csi + \kappa \in\R^d$ in a component along $M$ and a component in the
	orthogonal direction, and this has to be done in a way compatible with
	the lattice structure of $\Z^d$ and with the Floquet parameter. \\}
We consider the orthogonal decomposition
$\R^d=\Span_{\R}M\oplus (\Span_{\R}M)^{\perp}$. {Correspondingly, given a vector $w\in\R^d$,
we decompose it as}
$$
w=w_M+w_{M^\perp}\ ,\quad w_M\in\Span_{\R}M\ ,\quad w_{M^\perp}\in
(\Span_{\R}M)^{\perp} \ . 
$$

\begin{definition}
  \label{ortz}
  {Given a vector $\xi\in\Z^d$, a module $M$ and a Floquet parameter
  $\kappa$, 
  we define the
  following two objects:}
 \begin{equation} \label{betainz}
 \begin{gathered}
 \perpz\xi:=\xi-\inte{(\xi+\kappa)_M}\ ,\\
 \kappa':=\frazp{(\xi+\kappa)_M}\,.
 \end{gathered}
\end{equation}
\end{definition}
\begin{remark}
  \label{decompoz}
  If we denote $\zeta:=\inte{(\xi+\kappa)_M}$, one has
  \begin{equation}
    \label{dec.4}
(\xi+\kappa)_M=\zeta+\kappa'\ ,\quad
    (\xi+\kappa)_{M^\perp}=(\tilde \xi+\kappa)_{M^\perp}\ .
  \end{equation}
\end{remark}
{
Given a vector {$\beta\in \Z^d$,} we will have to consider the
space
\begin{equation}
\label{mbeta}
M+\beta:=\left\{ \xi\in \Z^d\ :\ \exists v\in M\ :
\ \xi=v+\beta \right\}\ .
\end{equation}
\begin{remark}\label{rmk costanti}
Notice that, for any $\csi
	\in  M + \beta,$ one has 
	$$
{	\tilde \csi  = \tilde \beta}\,, {	\quad \{(\csi + \kappa)_M\} = \{(\beta + \kappa)_M\}\,,}
	$$ 
{	thus the quantities $\tilde{\csi}$ and $\kappa^\prime$ defined in \eqref{betainz} are constant on $M + \beta$.}
\end{remark}
\begin{remark}
  \label{Gauge}
  The set $M + \beta$ defined as in \eqref{mbeta} is clearly an affine module isomorphic to
$M$.  A convenient way to identify the two spaces $M + \beta$ and $M$ is to subtract
  $\perpz\beta$ to a vector $w\in M+\beta$.\\
Correspondingly, the subspace of $L^2(\T^d)$ generated by
$M+\beta$ (in the sense of Definition \ref{def cal e}) is isomorphic to the subspace generated by
$M$. 
Explicitly, the isomorphism can be  realized by using the Gauge
transformation $\gaugemap$ defined by
\begin{equation}
  \label{gal}
\gaugemap u:=e^{-\ii \tilde{\beta} \cdot x}u\ . 
\end{equation}
\end{remark}
}
\begin{definition}
	\label{trasl}
	Given a module $M$, a vector {$\beta\in \Z^d$} and a set $W\subset
	M+\beta$, we denote $W^t:=W-\perpz \beta$ so that $\cW^t:=\gaugemap
	\cW\subset {L^2(\T^d)}$.
\end{definition}
{As a last step, we introduce the definitions of \emph{coordinates adapted to a module}.
	If $\bv^1,...,\bv^{d'}$ ($d'<d$) is a basis of $M\subset \Z^d$, then it
	can be completed to a basis of $\Z^d$, namely there exist
	$\bv^{d'+1},...,\bv^{d}$ such that the whole collection
	$\bv^1,...,\bv^{d}$ generates $\Z^d$. Such a basis will be called
	a basis {\it adapted to} $M$.
	In what follows, given a collection of such vectors $\{ \bv^{d^\prime + 1}, \dots, \bv^d\}$, we will denote
	\begin{equation}
	\label{mort}
	{\Mc}:={\rm span}_{\Z}\{\bv^{d'+1},...,\bv^{d} \}\ ;
	\end{equation}
	if $M=\Z^d$ then ${\Mc}=\{0\}$ and if $M=\{0\}$ then ${\Mc}=\Z^d$.
	{Of course, in general ${\Mc}$ is not unique, but this will not affect
		our construction.} Consider now the basis $\{\bu_j
        \}_{j = 1, \ldots, d}$ of $\R^d$ dual to $\{ \bv^j \}_{j =1 ,
          \ldots, d}$.}
\begin{definition}
  \label{adapted}
{        The
	coordinates $z^j$ introduced by
	\begin{equation}
	\label{cootd}
	x=z^j\bu_j
	\end{equation}
	are good coordinates on $\T^d$ (in the sense that they respect the
	$2\pi$ periodicity of the torus). These
	coordinates will be called \emph{coordinates adapted to $M$.}
}
\end{definition}

{The main result of this section is the following theorem. }

\begin{theorem}
\label{main}[Structure Theorem]
Given $\ep, \delta \in \R^+$ and
{$\tau > d - 1$ fulfilling
\begin{equation}\label{legami parametri}
\delta + d (d + \tau + 1) \ep < 1\,, \quad \ep (\tau + 1) \leq \delta\,,
\end{equation}
} a
Floquet parameter $\kappa$ and a flat metric $g$, there exists a partition of $\Z^d$:
\begin{equation}
\label{parti}
\Z^d= \bigcup_{M{\subseteq \Z^d}}\bigcup_{\beta\in \widetilde M}W_{M,\beta}
\end{equation}
where $M$ runs over the submodules of $\Z^d$ and $\widetilde M$ is a
subset of ${\Mc}$. {All the sets $W_{M, \beta}$ have finite cardinality}, the set $E_{\{0\}}:=\bigcup_{\beta}W_{\{0\}, \beta}$
has density one at infinity, and $W_{\Z^d,\{0\}}$ has cardinality bounded
by an integer $n_*$ 
which depends on the constants of the metric and on $d,\ \delta,\ \ep,
{\tau}$ only. 

Consider the operator \eqref{g}, assume that \change{$\cV\in\OPS^{\mu,\delta}$, $\mu < 2 \delta$, and let
	\begin{equation}\label{guadagno}
	\rho:= 2\delta - \mu\,;
	\end{equation}
	} then $\forall \tN> 0$ there exists a unitary transformation
$U$ which depends
smoothly on $\cV$, which fulfills
 \begin{align}
            \label{bohr}
	U - {\rm Id} \,,\ U^{-1} - {\rm Id} \in \change{\ops{\mu-\delta}}
	\end{align}
and is  s.t.
\begin{equation}
\label{tu}
UHU^{-1}=\widetilde H+\cR\ ,
\end{equation}
where
\begin{enumerate}
\item $\cR\in \change{\ops{\mu -\tN\rho}}$
\item $\widetilde{H}$ leaves invariant the subspaces generated by
  $W_{M, \beta}$ (according to Definition \ref{generated}) for all $M$
  and $\beta \in \widetilde M.$
  \change{\item Furthermore, one has the following:}
\begin{enumerate}
	\item[3.1]  $\forall \beta$, $\displaystyle{\widetilde{H}_{{\cal W}_{\{0\},
              \beta}}}\equiv \widetilde H\big|_{\cW_{\{0\} ,\beta}}$ is a Fourier multiplier
	\item[3.2] \change{If $\mu\leq 0$,} \label{lowerdim} $\forall M$ proper submodule and $\forall \beta \in
          \widetilde M$, one has that
	 {$H^{(1)}_{M,\beta}:=\gaugemap \widetilde{H}_{\cW_{M,\beta}} \gaugemap^*$} is a \emph{Schr\"odinger operator of
          dimension $d'=dim M\,,$} in the sense that introducing
          coordinates adapted to $M,$ it takes the form 
        \begin{equation}
          \label{onestep}
H^{(1)}_{M,\beta}= \Pi_{\cW^t_{M,\beta}} \left( - \Delta_{g, \kappa^\prime} + \cV_{M, \beta} {+ \norgs{(\tilde\beta+\kappa)_{M^\perp}}^2}\right) \Pi_{\cW^t_{M,\beta}}\,,
        \end{equation}
        where $-\Delta_{g, \kappa^\prime} $ is the $d'$ dimensional
        Laplacian computed with respect to the restriction of the
        metric $g^*$ to $\Span_{\R}M$ and with Floquet parameter
        $\displaystyle{\kappa^\prime = \{(\beta + \kappa)_M\}}$, and $\cV_{M, \beta} $ is a {periodic} pseudodifferential operator of order \change{$\mu$} (in $d'$
dimensions). \\
{Furthermore, the seminorms of the operators $U, {\cal R}$ and ${\cal
		V}_{M, \beta}$ only depend on the constants of the metric ${\frak c}, {\frak C},$ and on the seminorms of $\cal V.$}
\end{enumerate}

\end{enumerate}
\end{theorem}

\begin{remark}
  \label{decomp}
The partition of $\Z^d$ does not depend on the operator \eqref{g}, but
only on the properties of the metric, and on $\kappa$.  
  \end{remark}
\change{\begin{remark}\label{rmk m}
		Theorem \ref{main} is stated here in the more general case of an unbounded perturbation $\cV$ of order $\mu<2\delta$. However, we emphasize that Item 3. only holds true in the case of $\cV \in \ops{\mu}$ with $\mu \leq 0$. This is mainly due to the fact that, for positive $\mu$, $\cV \in \ops\mu$ does not imply $\cV_{M, \beta} \in \ops{\mu}$ as a $d^\prime$ dimensional operator with uniform boundedness of the seminorms with respect to $M, \beta$ (see Lemma \ref{prop riduzione} below). Thus an iterative application of Theorem \ref{main}, which is needed to prove our spectral result, requires to restrict to the case of bounded perturbations $\cV$.\\ For positive values of $\mu$, it is actually natural to expect that one could bound the seminorms of $\|\beta + \kappa\|^{\mu}\,\cV_{M, \beta}$ uniformly with respect to $M$ and $\beta$, with dependence only on the seminorms of $\cV$ and on the constants of the metric. However, we did not insist on pursuing such result here, as it goes beyond the interest of the present paper.
	\end{remark}
}

\begin{remark}
  \label{lowerset}
The theorem holds also if the initial operator \eqref{g} is replaced
by the restriction of a Schr\"odinger operator {to the subspace generated by any finite subset $E$ of
	$\Z^d$, with the only exception that in such a case the set $E_{\{0\}}$ does not have, of course, density one at infinity.} This is useful for iterating the construction.
\end{remark}
\begin{remark}
The restriction of the metric $g$ to a module $M$ has new constants which are controlled by the constants ${\frak c}$ and ${\frak C}$ of the initial metric $g\,.$ This is useful for the iteration of the construction.
\end{remark}


{Theorems similar to Theorem \ref{main} were proved in
 \cite{Par08,PS10,PS12} (see also \cite{PS09}). The main differences with the theorems proved in
 those papers are the following:}
 \begin{itemize}
 \item[1.] {The remainder $\cR:$ } {the kind of remainders obtained in those papers are
   not smoothing operators, but operators which are small when
   localized in an annulus in the action space $\xi$. 
 }

 \item[2.] {The presence of \change{Item 3.2}, {namely the dimensional reduction:} this is the main new point contained in
  Theorem \ref{main}. In particular
  the presence of this point allows to apply Theorem \ref{main} in an
  iterative way.
} 

\item[3.] {The last statement of the
  theorem, namely the uniformity of the constants with respect to the
  block: again this is needed in order to get an asymptotic expansion
  useful in order to study the time dependent case.}
 \end{itemize}

\section{The partition and the spectral theorem}\label{spec sec}

\subsection{Construction of the Partition}\label{geometrachilegge}

We are now giving the explicit construction of the sets $W_{M,\beta}$. This is
a quantum analogue of the classical geometrical construction of the
Nekhoroshev theorem {\cite{Nek1, Nek2} (see also \cite{GioPisa}). A direct classical counterpart can be found in
\cite{neknoi}. We found very
striking the fact that there is a so close connection between
classical and quantum dynamics.}
 {We remark that the construction of this section can also be
   considered as a variant of the construction of \cite{PS10, PS12}:
   the differences will be pointed out in the following.}

Roughly speaking, {given a submodule $M \subseteq \Z^d$ of dimension $s,$} the sets
\begin{equation}
  \label{BM}
{E^{(s)}_M}:=\bigcup_{\beta\in\widetilde M}W_{M,\beta}
\end{equation}
are  {the points} $\xi\in\Z^d$ which are resonant only with
the integer vectors of $M$. {For this reason, in the following we will often refer to a submodule $M \subseteq \Z^d$ ad a \emph{resonance module}.} In order to make
the construction precise, consider the classical symbol of
$-\Delta_{g,\kappa}$, namely
\begin{equation}
  \label{hzero}
h_0(\xi)=\left\|\xi+\kappa\right\|_{g^*}^2\ ;
\end{equation}
the frequencies of the corresponding classical motion are
$$
\omega_j=\xi_j+\kappa_j\ ,
$$
so that a point $\xi$ is (exactly) resonant with some integer $k$ if
$$
\scalgs{(\xi+\kappa)}{k}=0\ .
$$
Actually, the theory developed in \cite{noi} shows that, in a quantum
context a possible definition of point resonant with a vector $k$ is
\begin{equation}\label{def.res}
\left|\scalgs{(\xi+\kappa)}{k}\right|< \frac{\langle
	\xi+\kappa\rangle_{g}^\delta}{\|k\|_{g^*}^\tau} \ ,
\end{equation}
furthermore, due to the decay
of the Fourier coefficients of a smooth function, it is enough to
consider the $k$'s s.t.
$$
\|k\|_{g^*}\leq \langle\xi\rangle_{g}^\epsilon
$$
for some positive small $\epsilon$.
So, in principle ${E^{(s)}_M}$ should be the set of the $\xi$'s which are in resonance 
with the $k$'s belonging to $M$ and having a not too large
module. However this has to be modified due to the translation by
$k/2$ present in the definition of Weyl quantization. Furthermore, one
has to modify the construction both in order to get that the sets
${E^{(s)}_M}$ do not overlap and in order to obtain invariant sets.

To start with we define the resonance zones, in which the
following notation will be used:
\begin{definition} \label{csi traslate}
 Given $\csi \in \R^d$ and $k \in \Z^d,$ and a 	Floquet parameter
 $\kappa$, we define
        \begin{align}
          \label{traslo.1}
          \xik:=\xi+\kappa\ ,
          \\
          \label{traslo}
	\csi_{k} := \csi + \kappa +  \frac{k}{2}\equiv \xik+ \frac{k}{2}\,.
        \end{align}
\end{definition}
\begin{definition}[Resonant zones] \label{def res zones}
	Fix $\delta,\epsilon,\tau$ as in the statement of Theorem
        \ref{main}; fix also constants fulfilling:
	\begin{gather*}
	\delta_0 = \delta\,,\\
	\delta_{s+1} = \delta_s + {(d+ \tau + 1)\ep} \quad \forall s = 0, \dots, d-1\,,\\
	1 = D_0 < D_1 < \cdots D_{d-1}\,,\\
	1 = C_0 < C_1 < \cdots C_{d-1}\,,
	\end{gather*}
	then we define the following sets:
	\begin{enumerate} 
		\item
		$\displaystyle{ Z^{(0)} = \left \lbrace \csi \in
                  \Z^d\ \left|\ |\scalgs{\csi_k}{k}| > \langle \csi_k
                  \rangle_{g}^{\delta} {\norgs{k}^{-\tau}} \quad
                  \forall k \in \Z^{d}\ \textrm{ s. t. } \norgs{k} \leq
                  \langle \csi_k \rangle ^{\ep}_{g} \right. \right
                  \rbrace} $
		\item given $M \subseteq \Z^d$ a resonance module of
                  dimension $s\geq 1$
		and $s$ linearly independent vectors $\{ {k}_1, {k}_2, \dots, {k}_s \} \subset M$, we define
		\begin{multline}
		Z_{{k}_1, \dots, {k}_s} = \Big \lbrace \csi \in \Z^d
                \ \Big|\  \forall j = 1,\dots s \quad
                |\scalgs{\csi_{{k}_1}}{{k}_j}| \leq C_{j-1} \langle
                \csi_{{k}_1} \rangle ^{\delta_{j-1}}_{g}
                    {\norgs{{k}_j}^{-\tau}} \\ \label{r.z.1}
		\textrm{ and } \norgs{{k}_j} \leq D_{j-1} \langle
                \csi_{{k}_1} \rangle^{\ep}_{g} \Big \rbrace 
		\end{multline}
		and
		\begin{equation} \label{res zone sd}
		Z^{(s)}_{M} = \bigcup_{\begin{subarray}{c}
			\{{k}_1, \dots, {k}_s\} \subset M \\ \textrm{lin. ind. }
			\end{subarray}} Z_{{k}_1, \dots, {k}_s}\,.
		\end{equation}
	\end{enumerate}
\end{definition}

\begin{remark}
  \label{keneso}
By \eqref{r.z.1}, $\forall
s\geq1$ and $\forall M$, one has $Z^{(s)}_M\cap Z^{(0)}=\emptyset$.
\end{remark}

\begin{remark} \label{rmk zone inscatolate}
If $1\leq r<s$, then for any $M$ with dim $M=s$, one has
  $$
	Z^{(s)}_{M} \subseteq \bigcup_{ \begin{subarray}{c} M^\prime \subset M\\
		\textrm{dim.} M^\prime = r 
		\end{subarray}} Z^{(r)}_{M^\prime}\,.
	$$
\end{remark}

\begin{remark}
  \label{metrinoi}
{The fact that the zones $Z^{(0)}$ and $Z_{k_1,...,k_s}$ are defined
only in terms of the metric is one of the key ingredients allowing to
iterate the structure theorem.} 
\end{remark}
The regions $Z^{(s)}_{M}$ contain points $\csi \in \Z^d$ which are in
resonance with \textit{at least} $s$ linearly independent vectors in
$M$. Thus such regions are clearly not reciprocally disjoint. We identify
now the points $\csi \in \Z^d$ which admit \textit{exactly} $s$ linearly
independent resonance relations.
\begin{definition}[Resonant blocks] \label{def blocks}\ {Consider the following sets:}
	\begin{enumerate} 
	\item
		$$
		B^{(d)} := Z^{(d)}_{\Z^d}\,.
		$$
        \item Given $M \subset \Z^d$ a resonance module of dimension $s \in \{ 1, \dots, d-1\},$
		$$
		B^{(s)}_{M} := Z^{(s)}_{M} \backslash \left \lbrace \bigcup_{\begin{subarray}{c} M^\prime \textrm{s.t.} \dim M^\prime = s+1 \end{subarray}} Z^{(s+1)}_{M^\prime} \right \rbrace
		$$
		\item
		$$
		B^{(0)}:= Z^{(0)}
		$$
	\end{enumerate}
{We say that $B$ is a resonant block if $B = B^{(d)}, B = B^{(0)}$ or $B = B^{(s)}_M$ for some module $M$ of dimension $s$.}
\end{definition}
\begin{remark}
  \label{nonio}
The resonant blocks form a covering  of $\Z^d$.
\end{remark}
As proven below in Lemma \ref{lemma che bei blocchi} there exists a
suitable choice of the constants $C_s,\ D_s,\ \delta_s$ such that two blocks
$B^{(s)}_{M},\ B^{(s)}_{M^\prime} $ are disjoint if $M,\ M^\prime$ are
two distinct subspaces of equal dimension.

Still the blocks defined in Definition \ref{def blocks} do not provide
a suitable partition of $\Z^d,$ since  they are not left
invariant by the operator $\widetilde H$ of eq. \eqref{tu}.

Recall now that,
given two sets $A $ and $B,$  their Minkowski sum $A+B$ is defined by:
		$$
		A + B := \{ a + b \quad  | \quad a \in A\,,\ b \in B
                \}\ .
		$$
\begin{definition}[Extended blocks] \ 
	\begin{enumerate} 
		\label{def blocchi estesi}
\item $\displaystyle{E^{(0)}:=B^{(0)}\equiv Z^{(0)}}$
              \item Given a resonance module $M$ of dimension $1$, we define
		$$
		E^{(1)}_{M} := \left\{ B^{(1)}_{M} + M \right\} \cap Z^{(1)}_{M}\,,
		$$
		$$
		E^{(1)} := \bigcup_{M \textrm{ of dim. } 1} 
		E^{(1)}_{M}
		$$
		\item Given a resonance module $M$ of dimension $s$,
                  with 
                  {$2\leq s\leq d$,} 
                   we define
		$$
		E^{(s)}_{M} := \left\{ B^{(s)}_{M} + M \right\} \cap
                Z^{(s)}_{M} \cap 
                {\bigcap_{j=1}^{s-1} \left(E^{(s-j)}\right)^{c}\,,}
		$$
		{where, given $E \subseteq \Z^d,$ $E^c$ is the complementary set of $E$ in $\Z^d$.}
Correspondingly we define
		$$
		E^{(s)} := \bigcup_{ M \textrm{ of dim. }s } E
		^{(s)}_{M}\,.
		$$
	\end{enumerate}
\end{definition}
\begin{remark} \label{lemma ricoprimento}
The blocks $\{E^{(s)}_{M}\}_{M, s},\ E^{(0)},\ E^{(d)}\,$ form a
covering of $\Z^d$. {Actually, as shown in Theorem \ref{cor partizione} below, they form a partition of $\Z^d.$}
\end{remark}

It turns out that the decomposition $\Z^d=\bigcup_ME^{(s)}_M$ is
invariant for the operator $\widetilde H$ of Theorem
\ref{main}. Furthermore the sets  $E^{(s)}_M$ can still be decomposed
in invariant subsets which are given by
\begin{equation}
  \label{veri}
W_{M,\beta}:=E^{(s)}_M\cap(\beta+M)\ ,\quad 
\end{equation}
\begin{definition}
  \label{tildee}
The set of the $\beta\in {\Mc}$ s.t. the set \eqref{veri} is not empty
is denoted by $\widetilde M$.
\end{definition}

\begin{theorem}
  \label{specifica}
The sets $W_{M,\beta}$ of Theorem \ref{main} are the sets defined by
equation \eqref{veri}\ .
\end{theorem}

\subsection{Iteration and Spectral Theorem}\label{iterami}

\change{Given $\cV \in \ops{0}$,} Theorem \ref{main} allows to conjugate, up to smoothing operators, the operator \change{$H = -\Delta_{g, \kappa} + \cV$} to a sequence
of lower dimensional Schr\"odinger operators, the majority of which is
trivial (there are infinitely many Fourier multipliers and one finite
dimensional operator). In order to study the nontrivial Schr\"odinger
operators one can apply again Theorem \ref{main} to the operators of
eq. \eqref{onestep}. In this way one can conjugate each of these
operators to Schr\"odinger operators of lower dimension. Iterating
further and further, one is finally reduced to either finite
dimensional operators or Fourier multipliers.

\begin{remark}
  \label{direz}
The Schr\"odinger operators of eq. \eqref{onestep} act on
$\T^{d'}$ and the corresponding symbols, written in coordinates adapted to
$M$ depend only on the first $d'$ variables (both $x$ and $\xi$). If
one looks at such a symbol as the symbol of an operator on the original
torus, namely as a function in $\cC^{\infty}(T^*\T^d)$, then one has
that taking derivatives with respect to the $\xi$ variables does not
improve the decay in the directions of the variables which are not
present in the symbol, namely ($\xi^{d'+1},...,\xi^d$). For this
reason we will get that some eigenvalues (these are the unstable
eigenvalues of \cite{FKT2}) have asymptotics with only a directional decay.
\end{remark}

Directional decay is captured by the following definition which avoids
the introduction of adapted coordinates.

\begin{definition}
  \label{simdirez}
Let $m \leq 0,$ and let $M \subset \Z^d$ be a {proper} submodule, we say
that $a\in \cC^{\infty}(T^*\T^d)$ is a symbol of order $m$ in the
direction  $M$ if $\forall N_1, N_2 \in \N^{d}$ there exists a
constant $C_{N_1, N_2}>0$ such that
\begin{equation}
  \label{sidireq}
 \| d_x^{N_1} d_{\csi}^{N_2} a(x, \csi)\| \leq
C_{N_1, N_2} \langle (\csi + \kappa)_M \rangle_{g}^{m-\delta N_2} \quad \forall
x \in \T^d,\ \forall \csi \in \R^d\,.
\end{equation}
In this case we will write $a\in\simb m_M$. 
\end{definition}

\begin{definition}
\label{sim.1}
Given a module $M\subset \Z^d$,  a sequence of symbols 
$m_j\in \simb {-2j\delta}_{M}$, $j\geq0$, depending only on $\xi$ and a function
$m(\xi)$, possibly defined only on $\Z^d$ or on a subset $E$ of $\Z^d$, we
write 
\begin{equation}
\label{sim.eq}
m\simm{M} \sum_{j}m_j\ ,
\end{equation} 
if for any $N$ there exists $C_N$ s.t.
\begin{equation}
\label{sim.eq.2}
\left|m(\xi)-\sum_{j=0}^{N}m_j(\xi)\right|\leq\frac{
  C_N}{\langle (\xi+\kappa)_M\rangle_{g}^{(N+1)2 \delta}}\ .
\end{equation}
\end{definition}
\begin{definition}
  \label{ammissibile}
A sequence of modules
\begin{equation}
  \label{seq}
\Z^d\supset M^{(1)}\supset...\supset M^{(r)}\ ,\quad {\rm
  dim}M^{(j)}=d_j\ ,
\end{equation}
will be said to be admissible if
\begin{equation}
  \label{dj}
d_r\leq d_{r-1}<d_{r-2}<...<d_{1}\leq d\ ,
\end{equation}
and either $d_r=d_{r-1}$ or $d_r=0$ (namely the sequence ends when
either the last module coincides with the previous one or it consists
of $\{0\}$).

The number $r$ will be called the length of the sequence.

We will denote by $\mass$ the set of all admissible sequences of modules.
\end{definition}
We also denote $ \vec M:=(M^{(1)},...,M^{(r)}) $.

Let now $\vec M\in\mass$, then for any $j$ consider a module
${\left(M^{(j)}\right)^{(c)}}$ complementary to $M^{(j)}$ in $M^{(j-1)}$, namely a module
such that
$$
M^{(j)}+{\left(M^{(j)}\right)^{(c)}}=M^{(j-1)} \ ,\quad M^{(j)}\cap
{\left(M^{(j)}\right)^{(c)}}=\left\{0\right\}\ ,
$$
then the above construction forces to use also subsets
$$
\widetilde M^{(j)}\subset {\left(M^{(j)}\right)^{(c)}} \ .
$$
We denote
$$
\vec
M\widetilde{\ }:=(\widetilde M^{(1)},...,\widetilde M^{(k)})\ ,
$$ then the sequence of normalizations that one performs is determined
by the couple $(\vec M,\vec\beta)$ with $\vec\beta\equiv
(\beta_1,...,\beta_k)\in\vec M\widetilde{\ }$.

\begin{theorem}
  \label{speci_spettri}
  \change{Given $\cV \in \ops{0}$ and $H$ as in \eqref{g},} {there exists a bijective map
    \begin{equation}
      \label{mappa g}
\Z^d\ni\xi\mapsto \lambda_\xi\in \sigma(H)\ ,
    \end{equation}
(the eigenvalues being counted with multiplicity)    with the following properties:}
\begin{itemize}
  \item[(i)] {  there exists a constant $a\in(0,1)$, s.t. $\forall
    \tN\in\N$ there exists a constant $C_{\tN}$ s.t. the eigenfunction
    $\phi_\xi$ corresponding to $\lambda_\xi$ fulfills
    \begin{equation}
      \label{atima_soba}
\norma{\phi_\xi}_{H^{-\tN}}\leq \frac{C_{\tN}}{\lambda_\xi^{a\tN}}\ ,
    \end{equation}}

\item[(ii)]  {There exists a partition $$\Z^d=\bigcup_{\vec
  M\in\mass}\bigcup _{\vec\beta\in \vec M\widetilde {\ }}W_{\vec
  M,\vec \beta}\ , $$
and for any ${(\vec M, \vec \beta)}$ with $\vec M \in \mass$ and $\vec \beta \in \vec M\widetilde{\ } $ 
 there exists a sequence of $x$ independent symbols $\{m^{(j)}_{\vec
   M,\vec\beta}\}_{j \in \N},$ $m^{(j)}_{\vec M,\vec\beta} \in
 S^{-2\delta j, \delta}_{M^{(r-1)}}\ \forall j,$  with the following property. If $\xi \in W_{\vec
  M,\vec \beta}$, then
$\lambda_\xi$
    admits the asymptotic expansion
    \begin{equation}
      \label{asyfinal}
\lambda_\xi\mathop{\sim}^{M^{(r-1)}}\norgs{\xi+\kappa}^2+ \sum_{j \in \N} m^{(j)}_{\vec M,\vec\beta}(\xi)\ ,
    \end{equation}
    where $r$ is the length of the sequence $\vec M$. The operator $H$
    does not have other eigenvalues. Furthermore, the constants $C_N$
    of \eqref{sim.eq.2} are uniform with respect to the choice of the
    pair $(\vec M,\vec\beta)\,.$ } \end{itemize}
\end{theorem}

\begin{figure}
\centerline{\includegraphics[scale=0.6]{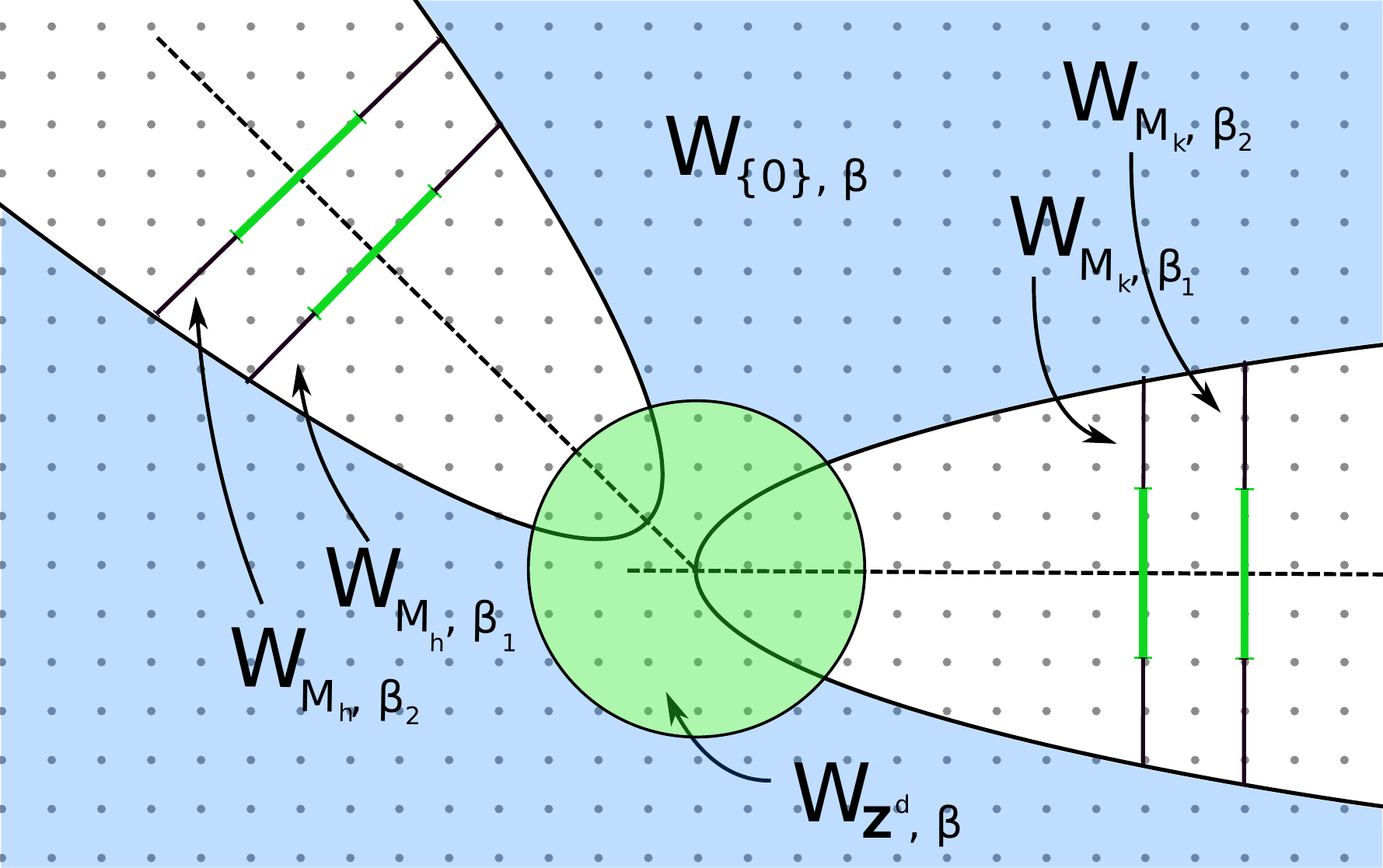}}
\caption{{ A cartoon of the block decomposition described in Theorem
    \ref{main}, in the case $d=2$, with $\kappa = 0$ and $g$ the
    Euclidean metric.
}}
\label{fig.reticolo}
\end{figure}
{The situation is illustrated in Figure \ref{fig.reticolo} in the
  case $d=2$, $\kappa = 0$ and $g$ the Euclidean metric. In such
  a case either $M = \Z^2,$ or $M = \{0\}$ or $M = M_k =
  \Span_\Z\{k\}$ for some $k \in \Z^2\,.$ In the figure only the
  resonant zones corresponding to $k=(0,1)$ and $k=(1,1)$ are plotted.
  In the blocks corresponding to
  $M = \{0\}$ (in blue) the normal form operator is a Fourier
  multiplier and here one gets the standard asymptotic expansion.  The
  block corresponding to $M = \Z^2$ (in green) has finite
  dimension. In the other blocks (white region) one can apply again
  Theorem \ref{main} getting finite dimensional sub-blocks (green
  segments) with uniformly bounded dimension plus trivial blocks (in
  the black segments). When adding the remainder one gets the
  directional asymptotics with decay in the direction of the
  segments.  }

\vskip25pt

\Large  \noindent{\bf PART II: Proofs}
\normalsize

\section{Quantum normal form} \label{sez forma normale nonris}

In this section we give a variant of the normal form construction of
\cite{noi} {and \cite{\russi} suitable for our goal.}
 From
now on we will drop the index $g$ or $g^*$ from the notation of the
scalar products and of the norms. Of course the scalar products of
vectors will be computed using the metric $g_{AB}$ and those of
covectors using the metric $g^{AB}$. In particular the scalar products
involved in the definitions of resonance are always between covectors.

{
	\begin{definition} \label{def resonant op}
	A symbol $ \displaystyle{\nform(x, \csi) = \sum_{k \in \Z^d}
          \hat{\nform}_k(\csi) e^{\ii k \cdot x}}\change{\in\simb m, m \in \R,} $ is said to be
        in (resonant) normal form if, $\forall k\in\Z^d \backslash \{ 0 \}\,,$
		\begin{multline} \label{potenziale risonante}
		\textrm{supp}( \hat{\nform}_k) \subseteq \bigg \lbrace \csi \in \R^d \ \bigg|\ \left|\scala{\csi + \kappa}{k} \right| \leq \left \langle \csi + \kappa
		 \right \rangle^{\delta} {\norm{k}^{-\tau}}\ \textrm{and} \quad \norm{k} \leq \langle
                 \csi + \kappa \rangle^{\ep}  \bigg
                 \rbrace \ .
		\end{multline}
\end{definition}
}
\begin{definition} Let $M\subset \Z^d$ be a module, then a symbol
  $\nform \in\change{\simb m}$ is said to be in \emph{normal form with respect to
    $M$} if it is in normal form and furthermore its Fourier transform
  is given by 
	\begin{equation}
	\nform(x, \csi) = \sum_{k \in M} \hat{\nform}_k (\csi) e^{\ii k \cdot x}\,.
	\end{equation}
\end{definition}
\begin{definition}
  \label{norop}
A pseudodifferential operator will be said to be in normal form
(resp. normal form with respect to a module $M$) if the
corresponding symbol is in normal form (resp. normal form with respect to a
module $M$).
  \end{definition}

\begin{lemma} \label{rmk reg 1d} In dimension one (namely if $d =1$)
  operators in normal form are smoothing, namely of class
$\ops{-\infty}$.
\end{lemma}
\proof
Let ${\NForm = \Op(\nform)}
$ {be in normal form}
then, $\forall k \in \Z$
	$$
	\hat{\nform}_k(\csi) \neq 0 \Rightarrow |\scala{\csi + \kappa}{k}| \leq \langle \csi + \kappa \rangle^{\delta} {\| k\|^{-\tau} \leq \langle \csi + \kappa \rangle^{\delta} {\frak c}^{-\frac{\tau}{2}}}\,.
	$$
Since $\csi + \kappa \parallel k\,,$ it follows that
	\begin{align*}
	\norm{\csi + \kappa} & \leq {\frak c}^{-\frac{1}{2}} \norm{k} \norm{\csi + \kappa}\\
	& = {\frak c}^{-\frac{1}{2}} |\scala{\csi + \kappa}{k}| \\
	& \leq {{\frak c}^{-\frac{(1 + \tau)}{2}}} \langle \csi + \kappa \rangle^{\delta}
	\end{align*}
	which, by $\delta < 1\,,$ implies the existence of a
        constant $C$ such that $\|\csi +
        \kappa\| \leq C\,.$ 
 \qed

\begin{definition}\label{def average}
Given a symbol $a$ we define its average by
$$
\langle a \rangle (\xi) = \frac{1}{\mu_g(\T^d)}\int_{\T^d} a(x, \csi) \ d \mu_g(x)\,.
$$
If $ A = \Op(a)$, we denote $\langle A \rangle (D) = \Op(\langle a \rangle (\xi))$. 
\end{definition}

The following result is just a small modification of Theorem 5.1 of
\cite{noi} {(which in turn is a variant of Theorem 4.3 of \cite{PS10} and Theorem 9.2 of \cite{PS12}.)}.  
\begin{theorem} \label{ regolarizzaz non ris}
	Consider the operator
	$H = -\Delta_{g, \kappa}+  {\cal V}\,, $ with
${\cal V} = \Op(V) \in \change{\OPS^{\mu, \delta}\,, \mu < 2 \delta,}$ \change{and let $\rho$ as in \eqref{guadagno}.}
For all $\tN >0$ there exists a unitary transformation $U_\tN$ such
that
\begin{itemize}
\item[1)]         \begin{align}
            \label{boh}
	U_\tN - {\rm Id} \,,\ U_\tN^{-1} - {\rm Id} \in \change{\OPS^{\mu-\delta, \delta}}
        \\
\label{bah}
        U_\tN H U^{-1}_\tN= {\cal L}^{(\tN)} = \tilde{H}^{(\tN)} + {\cal R}^{(\tN)}
	\end{align}
with $\displaystyle{{\cal R}^{(\tN)} \in \change{\ops{\mu-\tN\rho}}}$, and{
\begin{equation}
  \label{HN}
\tilde{H}^{(\tN)} = -\Delta_{g, \kappa} + {\NForm}^{(\tN)}\,,
  \end{equation}
where $ {\NForm}^{(\tN)} \in \change{\OPS^{\mu, \delta}}$ is in resonant normal
form.\\ {
{Furthermore, \change{in the case $\mu \leq 0$} the families of seminorms of the operators ${\cal N}^{(\tN)}, {\cal R}^{(\tN)}, U_\tN$ only depend on the family of seminorms of the operator $\cV$ and on the constants of the metric, as well as on $\tN,$ on $d$ and on the parameters $\delta, \ep, \tau\,.$ }
}}
\item[2)] Let $E\subset\Z^d$ be a subset and let ${\cal E}$ be the
  space it generates according to \eqref{def cal e}. If ${\cal V}$
  leaves ${\cal E}$ invariant, namely $[{\cal V}, \Pi_{\cal E}] =
  0\,,$ then one has
	\begin{equation}\label{trasf commutano}
	[U_\tN\,, \Pi_{\cal E}] = 0 \,.
	\end{equation}
\end{itemize}
\end{theorem}
\begin{remark}\label{rmk proiettore}
	Assume that $\cV$ leaves invariant a subspace $\cal E$ of the form  \eqref{def cal e}. Then, by Item $2$ of Theorem \ref{ regolarizzaz non ris}, one also has
	$$
	U_\tN \Pi_{\cal E} H \Pi_{\cal E} U_\tN^{-1} = \Pi_{\cal E} \cL^{(\tN)} \Pi_{\cal E}\,.
	$$
\end{remark}


{The proof of Theorem \ref{ regolarizzaz non ris} is a small variant
  of the proof of Theorem 5.1 of \cite{noi} (and of Theorem 4.3 of \cite{PS10}),
  so here it will  only be sketched.}

The proof is obtained working
at the level of the symbols and is based on a decomposition that we
now recall. First consider an even function $\chi: \R \rightarrow [0,
  1]$ with the property that {$\chi(t) = 1$} for all $t$ with $|t| \leq
\frac{1}{2}$ and {$\chi(t) = 0$} for all $t$ with $|t| \geq 1.$ 

\begin{definition}\label{rmk dec symb}
Given $\ep, \delta > 0$ and $\tau > d-1$ as in Theorem \ref{main},  define the following functions:
\begin{align*}
\chi_k(\csi) = \chi\left(\frac{2 \|k\|^\tau (\csi + \kappa, k) }{\langle
  \csi + \kappa \rangle^\delta}\right)\ , \quad  k &\in \Z^d\backslash
\{0\},
\\
\tilde{\chi}_k(\csi) = \chi \left(\frac{\|k\|}{\langle \csi + \kappa
  \rangle^\ep}\right)\ , \quad k &\in \Z^d\backslash \{0\}\ .
\end{align*}
Correspondingly, given a symbol $w \in \change{S^{m,\delta}}$, we decompose
it as follows:
\begin{equation}\label{w a pezzi}
w = \langle w \rangle +  w^{(\rm nr)} + w^{(\rm res)} + w^{(S)}\,,
\end{equation}
where $\langle w \rangle$ is the average symbol of $w\,,$ and
\begin{equation*} 
\begin{gathered}
w^{(\rm res)}(x, \csi) = \sum_{k \neq 0} \hat{w}_k(\csi) \chi_k(\csi) \tilde{\chi}_k(\csi) e^{\ii k \cdot x}\,,\\
w^{(\rm nr)}(x, \csi) = \sum_{k \neq 0} \hat{w}_k(\csi) \left(1 - \chi_k(\csi)\right) \tilde{\chi}_k(\csi) e^{\ii k \cdot x}\,,\\
w^{(S)}(x, \csi) = \sum_{k \neq 0} \hat{w}_k(\csi) \left(1 - \tilde{\chi}_k(\csi)\right) e^{\ii k \cdot x}\,.
\end{gathered}
\end{equation*}
\end{definition}

As proved in Lemma 5.6 of \cite{noi}, the functions $\langle w
\rangle\,, w^{(\rm nr)}\,, w^{(\rm res)}$ are symbols of the same
order of $w$ while $w^{(S)} \in \simb{-\infty}$. 

\begin{proof}[Sketch of the proof of Thm \ref{ regolarizzaz non ris}]
  Following \cite{noi}, 
  Part 1 is proved  iteratively: consider an operator of the form 
$$
H_j = -\Delta_{g, \kappa} + {\NForm}^{(j)} + {\cal R}^{(j)}\,,
$$
with ${\NForm}^{(j)} \in \change{\ops{\mu}}$ in resonant normal form and ${\cal
  R}^{(j)} \in \change{\OPS^{\mu-j\rho, \delta}.}$ We look for a pseudodifferential operator
$\cG_j$ such that 
$$
e^{\ii {\cal G}_j} H_j e^{-\ii {\cal G}_j} = H_{j+1} .
$$
To this end remark that, by standard pseudodifferential calculus, one
has
$$
e^{\ii {\cal G}_j} H_j e^{-\ii {\cal G}_j} = -\Delta_{g,\kappa}{-}\im
[-\Delta_{g,\kappa};\cG_j ]+
\NForm^{(j)}+\cR^{(j)}+\text{lower\ order\ terms}\ 
$$
\change{(see also the results collected in Appendix \ref{pdc}).}
If $G_j$ is the symbol of $\cG_j$ and $R^{(j)}$ the symbol of $\cR^{(j)}$,
then the symbol of $\im
[-\Delta_{g,\kappa};\cG_j ]+\cR^{(j)}$ is, using the decomposition of
Definition \ref{rmk dec symb}, 
\begin{align*}
\left\{ \|{\xi+\kappa}\|^2;G_j  \right\}+  \langle R^{(j)} \rangle +  (R^{(j)})^{(\rm nr)} + (R^{(j)})^{(\rm res)} + (R^{(j)})^{(S)}\,.
\end{align*}
This is in normal form up to smoothing terms if $G_j$ is chosen in
such a way that the following equation holds:
$$
\{ \|\csi + \kappa\|^2,\ G_j \} + (R^{(j)})^{(\rm nr)} = 0\ .
$$
This equation is fulfilled if $G_j$ is defined by
\begin{equation}\label{g eq}
G_j(\csi, x) = - \sum_{k \neq 0} \frac{\widehat{\left((R^{(j)})^{\rm
      (nr)}\right)}_{k}(\csi) }{ 2 i (\xi+\kappa)  \cdot k} e^{\ii k \cdot x}\,.
\end{equation}
Using such a $G_j$ to generate the corresponding unitary
transformation and iterating, one gets the proof of part 1 of the
Theorem.

To prove part 2, namely the commutation relation \eqref{trasf
  commutano}, we proceed inductively. First of all we observe that, in
the case $j = 0,$ ${\NForm}^{(0)} = 0$ and
$$
[{\cal R}^{(0)}, \Pi_{\cal E}] = [{\cal V}, \Pi_{\cal E}] = 0\,.
$$
Let us now fix some $j \geq 0$ and suppose that ${\NForm}^{(j)}$ and
${\cal R}^{(j)}$ commute with $\Pi_{\cal E}$.

Given a self-adjoint operator $A,$ since ${\cal E}$ has the form
$\displaystyle{{\cal E} = \overline{\Span \{e^{\ii k \cdot x}\ |\ k \in E\}}}$,
the condition
$[A, \Pi_{\cal E}] = 0$ holds if and only if
\begin{equation}\label{v commuta}
\left(k \in E\ , \quad A_{k}^{k^\prime} \neq 0\right) \Rightarrow k^\prime \in E\,,
\end{equation}
where 
$\displaystyle{A_{k}^{k^\prime} = \frac{1}{\mu_g(\T^d)}\langle A e^{\ii k \cdot x},\ e^{\ii
    k^\prime \cdot x} \rangle}$ are the matrix elements of
$A$ with respect to the basis of the Fourier modes.  Furthermore, by
definition of Weyl quantization one has that, if $A = \Op(a)\,,$
\begin{equation}\label{m_el_weyl}
A_{k}^{k^\prime} = \hat{a}_{k^\prime - k}\left( \frac{k +
  k^\prime}{2}\right)\ .
\end{equation}
Due to definitions \eqref{w a pezzi} of the symbols $(R^{(j)})^{\rm (nr)}$ and $(R^{(j)})^{\rm (res)},$ equation \eqref{m_el_weyl} immediately implies that 
\begin{gather*}
\left(\Op((R^{(j)})^{\rm (nr)})\right)_{k}^{k^\prime} \neq 0 \,, \textrm{ or }\  \left(\Op((R^{(j)})^{\rm (res)})\right)_{k}^{k^\prime} \neq 0\,  \textrm{ for some } k, k^\prime \in \Z^d\,,\\  \Longrightarrow ({\cal R}^{(j)})_{k}^{k^\prime} \neq 0\,,
\end{gather*}
Similarly, $$
\left({\cal G}_j\right)_{k}^{k^\prime} \neq 0 \,  \quad  \Longrightarrow ({\cal R}^{(j)})_{k}^{k^\prime} \neq 0\,.
$$
This, together with condition \eqref{v commuta}, enables to conclude that ${\cal G}_{j}$ commutes with $\Pi_{\cal E}\,,$ and so do $\Op\left((R^{(j)})^{\rm (res)}\right)$ and $\Op\left((R^{(j)})^{\rm (nr)}\right)\,.$ Hence $e^{-\ii {\cal G}_j}$ commutes with $\Pi_{\cal E},$ since ${\cal G}_j$ does. The same holds for ${\NForm}^{(j+1)},$ since
$$
[{\NForm}^{(j+1)}, \Pi_{\cal E}] = [{\NForm}^{(j)}, \Pi_{\cal E}] + [\Op\left(\langle R^{(j)} \rangle\right), \Pi_{\cal E}] + [\Op\left((R^{(j)})^{\rm (res)}\right), \Pi_{\cal E}] = 0\,,
$$
and
$$
{\cal R}^{(j+1)}  = e^{\ii {\cal G}_j} H_j e^{-\ii {\cal G}_j} -
\left( -\Delta_{g, \kappa} + {\NForm}^{(j+1)} \right)\ .
$$
\end{proof}

\section{Geometric part} \label{sez blocchi}

In order to iterate Theorem \ref{main} we will have to work in a
subspace of $L^2$ generated by some subset 
$E\subset \Z^d$. From now on we will {fix $E \subseteq \Z^d$ and  develop all the proofs taking
$\xi$ as a variable in $E$ .
Accordingly to this, we will replace the extended blocks $E^{(s)}_{M}$ of Definition \eqref{def blocchi estesi} with $E^{(s)}_{M} \cap E$ , which we still denote by $E^{(s)}_{M}\,.$ We will do the same for the blocks $B^{(s)}_{M}$ and for the zones $Z^{(s)}_{M}$ of Definitions \ref{def res zones}, \ref{def blocks}.}

\subsection{Properties of the extended blocks $E^{(s)}_{M}$, non
  overlapping of resonances.}

We show here that the extended blocks 
$E^{(s)}_{M}$ form a partition {of $E $} and prove some properties which are
needed in order to show that they are left invariant by an operator in
normal form. 
As in the proof of the classical Nekhoroshev Theorem, the following
Lemma plays a fundamental role.

\begin{lemma}\label{giorgilli metrica g}
	Let $s \in \{1\,, \dots\,, d\}$ and let $\{ u_1\,, \dots u_s
        \}$ be linearly independent vectors in $\R^d\,.$ Let $w \in
        \Span{\{ u_1\,, \dots u_s\}}$ be any vector. If $\alpha\,, N$
        are such that
	\[
	\begin{gathered}
	\norm{u_j} \leq N \quad \forall j = 1\,, \dots s\,,\\
	| \scal{w}{u_j} | \leq \alpha \quad \forall j = 1\,, \dots s\,,
	\end{gathered}
	\]
	then
	\[
	\norm{w} \leq \frac{s N^{s-1} \alpha }{\textrm{Vol}_g\{u_1\,| \cdots\,| u_s \}}\,.
	\]
\end{lemma}
This is just a coordinate free formulation of Lemma 5.7 of
\cite{GioPisa}, which is recalled in the appendix as Lemma
\ref{giorgilli tipo kramer}.  By \eqref{volmin}, one also has that, if
$u_j\in\Z^d$, $\forall j=1,...,s$, then
\begin{equation}
  \label{c.storta}
\norm{w} \leq {s N^{s-1} \alpha {\frak C}^{-1}}  \,.
\end{equation}

We state now a couple of simple properties of the extended
blocks.

\begin{lemma} \label{lemma bd piccolo}
	The extended block $E^{(d)}$ is finite dimensional; in particular, there exists a positive $n_* = n_*({\frak c}, {\frak C}, \ep,{\tau},  \delta_{d-1}, C_{d-1}, D_{d-1})$ such that
	$$
	E^{(d)} \subseteq \left \lbrace \csi \in \R^d \ \left|\ \norm{\csi + \kappa} \leq  n_*\right. \right \rbrace\,.
	$$
\end{lemma}
\begin{proof}
	If $\csi \in E^{(d)},$ in particular there exist $\{{k}_1, \dots, {k}_{d}\} \subset \Z^d$ linear independent vectors such that
	$$
	\begin{gathered}
	\norm{k_1} \leq D_0 \langle \csi_{k_1} \rangle^{\ep}\,,\\
	\norm{{k}_j}\leq D_{j-1} \langle \csi_{k_1} \rangle ^{\ep} \leq D_{d-1} \langle \csi_{k_1} \rangle ^{\ep}\,,\\
	| \scala{\csi_{k_1}}{{k}_j}| \leq C_{d-1} \langle \csi_{k_1} \rangle ^{\delta_{d-1}}{\norm{{k}_j}^{-\tau}} \,.
	\end{gathered}
	$$ In order to eliminate the indexes $k_1$ from $\xi$, we
        apply Lemma \ref{xixik}, with $\sigmav = \eta = \csi^\kappa,$
        $l = 0,$ $h = k_j$ and $k = \frac{k_1}{2}$ to
        deduce that there exist constants
	$$
	\Cip = \Cip({\frak c}, \ep, {\tau}, \delta_{d-1}, D_{d-1}, C_{d-1})\,, \quad \Dip = \Dip({\frak c}, \ep, {\tau}, \delta_{d-1}, D_{d-1}, C_{d-1})
	$$
	such that
	$$
	|\scala{\csi^\kappa}{k_j}| \leq \Cip \langle \csi^\kappa \rangle^{\delta} \norm{k_j}^{-\tau}\,, \quad \norm{k_j} \leq \Dip \langle \csi^\kappa \rangle^\ep\,.
	$$
	Recalling that ${\frak c}$ is such that, for all $h \in \Z^d,$ $\norm{h}^2 \geq {\frak c}\,$
	and using Lemma \ref{giorgilli metrica g},  and
        Eq. \eqref{c.storta} we have 
	\[
	\norm{\csi^\kappa} \leq d {\frak c}^{-\tau/2}{{\frak C}^{-1}} {C}^\prime\ {(D^\prime)}^{d-1} \langle \csi^\kappa \rangle^{d \ep + \delta}\,,
	\]
	which, applying Remark \ref{B.11.1} with $a = d \ep + \delta < 1$, implies the existence of a constant $n_* = n_*(\delta, \ep, \tau, {C}^\prime, {D}^\prime, {\frak c}, {\frak C})$ such that $\norm{ \csi^\kappa} < \bar{N}\,.$
\end{proof}
{
\begin{lemma}\label{is big}
	{If $E = \Z^d,$} the set $E^{(0)}$ is of density one at infinity, namely
	$$
	\lim_{R \rightarrow \infty} \frac{ \sharp \left(E^{(0)} \cap B_R (0) \right) }{\sharp \left(\Z^d \cap B_R(0)\right) } = 1\,.
	$$
\end{lemma}
\begin{proof}
	We exploit the fact that a set is of density one at infinity
        if and only if its complementary set is of density zero, and
        we analyze the complementary set of $E^{(0)}$. Recall
        that ${E^{(0)} = Z^{(0)}}$ so that, by Definition \ref{def res zones}, its complementary set is
        \begin{align*}
          \Z^d \backslash E^{(0)}&= \bigcup_{ M \textrm{ of dim. }1 }
          Z^{(1)}_M \\
 &         = \{ \csi \in \Z^d \ |\ \exists k \in \Z^d \textrm{ s. t. }  |\scala{\csi_k}{k}|\leq \langle \csi_{k} \rangle^{\delta}\norm{k}^{-\tau}\,,\ \norm{k}\leq \langle \csi_{k} \rangle^{\ep}\}\,.
	        \end{align*}
	
By Lemma \ref{xixik} there exists constants $C^\prime, D^\prime$ depending only on $\delta, \ep, \tau, {\frak c}, {\frak C}$ such that
	$$
	\begin{gathered}
	\Z^d \backslash E^{(0)} \subseteq \{ \csi \in \Z^d \ |\ \exists k \in \Z^d \textrm{ s. t. }  |\scala{\csi}{k}|\leq C^\prime \langle \csi \rangle^{\delta}\norm{k}^{-\tau}\,,\ \norm{k}\leq D^\prime \langle \csi \rangle^{\ep}\}\,.
	\end{gathered}
	$$
	But the latter is the complementary set to
	$$
	\Omega = \{ \csi \in \Z^d\ |\ |\scala{\csi_k}{k}|> C^\prime \langle \csi_{k} \rangle^{\delta}\norm{k}^{-\tau} \quad \forall k \in \Z^d\ \textrm{ s. t. } \norm{k}\leq D^\prime \langle \csi_{k} \rangle^{\ep}\, \}\,.
	$$
Then Proposition 5.9 of \cite{noi} gives the result.
\end{proof}
}

We now analyze the other blocks.

First remark that, if $s^\prime \neq s,$ then two extended blocks
$E^{(s)}_{M}$ and $E^{(s^\prime)}_{M^\prime}$
are disjoint. Then we have to prove that two different
extended blocks of the same dimension do not intersect. To this end   a
further geometric analysis is required. 

\begin{lemma} \label{lemma piccole proiez}
If $\csi \in Z^{(s)}_M$ then there exists a positive constant $K$ depending only on ${\frak c}, {\frak C}, d, \ep, {\tau}, \delta_{s-1}, C_{s-1}, D_{s-1},$ such that 
\begin{equation}
  \label{proi.1}
\norm{(\xik)_M} \leq K \langle \xik \rangle^{\delta_{s-1}+d\epsilon}\,.
\end{equation}
\end{lemma}
\begin{proof}
Since $\csi \in Z^{(s)}_{M}\,,$ there exist $\{{k}_1, \dots, {k}_s\} \subset M$ linearly independent vectors such that for all $j = 1, \dots, s$
\begin{equation}
  \label{p.p.1}
 \begin{gathered}
|\scala{\left(\csi_{{k}_1}\right)_M}{{k}_j}| = |\scala{\csi_{{k}_1}}{{k}_j}|\leq C_{j-1} \langle \csi_{{k}_1} \rangle^{\delta_{j-1}} {\| {k}_j\|^{-\tau}}\,, \quad
\norm{{k}_j} \leq D_{j-1} \langle \csi_{{k}_1} \rangle^{\ep}\,.
\end{gathered}
\end{equation}
Then, by Lemma \ref{xixik} one can substitute in the above
formulae $\xik$ to $\xi_{ k_1}$; precisely, there exist two
positive constants $\Cip,\ \Dip = \Cip, \Dip({\frak c}, \ep,
\delta_{s-1}, C_{s-1}, D_{s-1}),$ such that, 
\begin{align*}
\begin{gathered}
|\scala{(\xik)_M}{{k}_j}| =|\scala{\xik}{{k}_j}| \leq \Cip
\langle \xik 
\rangle^{\delta_{s-1}} {\norm{{k}_j}^{-\tau} \leq C^\prime {\frak c}^{-\tau/2} \langle \xik 
\rangle^{\delta_{s-1}} }\,,\\
\norm{{k}_j} \leq \Dip \langle \xik \rangle^{\ep}\,.
\end{gathered}
\end{align*}
By Lemma \ref{giorgilli metrica g}, there exists $C = C(d)$ such that
$$
\begin{aligned}
\norm{(\xik)_M} &\leq C(d) \frac{ (\Dip)^{d}\langle \csi^\kappa
	\rangle^{d \ep} }{\textrm{Vol}_g({k}_1 | \cdots |
	{k}_s)} \Cip {{\frak c}^{-\tau/2}} \langle \xik \rangle^{\delta_{s-1}}\,,
\end{aligned}
$$
and therefore, {recalling that $\textrm{Vol}_g({k}_1 | \cdots |
{k}_s) \geq {\frak C}$ (see the definition of $\frak C$ as in \eqref{volmin}),} the thesis holds. 
\end{proof}

By definition, the points belonging to a block ${B}^{(s)}_M$ are
resonant only with vectors $k \in M$. A priori,
this property does not hold true for points in the extended block
$E^{(s)}_M\,.$ So we need an estimate of the distance between
$E^{(s)}_M$ and $B^{(s)}_M\,.$
\begin{lemma} \label{lemma diam}
	Let $\delta_{s-1} + d\ep < 1\,$ and $M$ with $dim M=s$; if
        $\zeta \in E^{(s)}_M$ then there exists $\csi \in
        B^{(s)}_{M}$ and a positive constant $F$ depending only on
        ${\frak c}, {\frak C}, d, \ep, {\tau}, \delta_{s-1}, C_{s-1}, D_{s-1}$
        such that
	\begin{equation}
          \label{diam.1}
\norm{\csi - \zeta} \leq F \langle \xik \rangle^{\delta_{s-1} + \ep
  d}\ ,\quad \norm{\csi - \zeta} \leq F \langle \zetak \rangle^{\delta_{s-1} + \ep
  d}
        \end{equation}
\end{lemma}
\begin{proof}
 If $\zeta \in E^{(s)}_{M},$ then in particular $\zeta \in
 Z^{(s)}_M$ and there exists a point $\csi \in B_{M}^{(s)}$ such that
 $\zeta = \csi + \upsilon,$ with $\upsilon \in M.$ In particular,
 $(\csi)_M^{\bot} = (\zeta)_M^{\bot}\,,$ hence one has
	\begin{align*}
	\norm{\csi - \zeta} = \norm{\left(\csi - \zeta\right)_{M}}
	\leq \norm{(\xik)_M|} + \norm{(\zetak)_M}\,.
	\end{align*}
	Since $\csi \in Z^{(s)}_M$ and $\zeta \in Z^{(s)}_M$, due to
        Lemma \ref{lemma piccole proiez}, there exists $K$, such that
	\begin{equation} \label{stime proiez}
	\norm{(\xik)_{M}} \leq K \langle \xik \rangle^{d \ep +
          \delta_{s-1}}\,, \quad \norm{(\zetak)_{M}} \leq K
        \langle \zetak \rangle^{d \ep + \delta_{s-1}}\,.
	\end{equation}
Exploiting Remark \ref{B.11.1} with $a=\delta_{s-1}+\epsilon d$, one gets 
$$
\langle\zetak\rangle^a=\langle\xi+\kappa+\zeta-\xi\rangle^a\leq
K'(\langle\xik\rangle^a +\norm{\zeta-\xi}^a)
$$
and, exploiting Lemma \ref{xallaa}, we immediately get
$$
\norm{\zeta-\xi}\leq F\langle\xik\rangle^a\ .
$$
Inverting the role of $\xi$ and $\zeta$ one gets the other
estimate.
\end{proof}

The next two lemmata ensure that, if the parameters $C_{j}, D_j$ are
suitably chosen for all $j$, an extended block $E^{(s)}_M$ is far from
every resonant zone associated to a lower dimensional module
$M^\prime\,$ which is not contained in $M.$ 
\begin{lemma} \label{lemma che bei blocchi}[Non overlapping of resonances]
	For all $s = 1, \dots d-1$ there exist positive constants $\bar{C}_s$ and $\bar{D}_s$, depending only on ${\frak c}, {\frak C}, d, C_{s-1}, D_{s-1}, \ep, \delta_{s-1}, {\tau}\,,$
	such that the following holds: suppose that $M$ and $M^\prime$
        are two distinct resonance modules of respective  dimensions  $s$ and $s^\prime$ with $s^\prime \leq s$ and  $M^\prime \nsubseteq M.$ If
	$$
	C_s > \bar{C}_s\,,\quad  D_s >\bar{D}_s\,,
	$$
	then
	$$
E^{(s)}_{M} \cap Z^{(s^\prime)}_{M^\prime} = \emptyset\,.
	$$
\end{lemma}
\begin{proof}
Assume by contradiction
that there exists $\zeta \in E^{(s)}_{M} \cap
Z^{(s^\prime)}_{M^\prime},$ then there exists $\xi\in B^{(s)}_M$
s.t. \eqref{diam.1} holds.

Since $\zeta \in Z^{(s')}_{M'}$, there exist $s'$
integer vectors, $k_1,...,k_{s'}\in M'$ {\it among which at least one does not
  belong to $M$} s.t.
\begin{equation}
  \label{bei.1}
\left|\scala{\zeta_{k_1}}{k_j}\right|\leq
C_{j-1}\langle\zeta_{k_1}\rangle^{\delta_{j-1}} {\norm{k_j}^{-\tau}}\ ,\quad \norm{k_j}\leq
D_{j-1}\langle\zeta_{k_1}\rangle ^{\epsilon}\ . 
\end{equation}
Let $k_{\bar\j }$ be the vector which does not belong to $M$; the idea
is to show that the resonance relation of $\zeta$ with $k_{\bar\j }$
implies an analogous relation for $\xi$, but this will be in
contradiction with the fact that $\xi\in B^{(s)}_M$ (which contains
vectors which are {\it only} resonant with $M$).

To start with remark that, since $\xi\in B^{(s)}_M\subset Z^{(s)}_M$,
there exist $l_1,...,l_s\in M$, linearly independent, s.t.
\begin{equation}
  \label{bei.4}
\left|\scala{\xi_{l_1}}{l_j}\right|\leq
C_{j-1}\langle\xi_{l_1}\rangle^{\delta_{j-1}} {\norm{l_j}^{-\tau}}\ ,\quad \norm{l_j}\leq
D_{j-1}\langle\xi_{l_1}\rangle ^{\epsilon}\ . 
\end{equation}
We now apply Lemma \ref{xixik} with $h:=k_{\bar\j }/2$, $\ell:=l_1/2$,
$\sigmav:=\zeta+\kappa$, $\eta:=\xi+\kappa$. So, \eqref{xixi.5}
implies
$$
\left|\scala{\xi_{l_1}}{k_{\bar\j }}\right|\leq
K'{\langle\xi_{l_1}\rangle^{\delta_{s-1}+\epsilon(d+\tau+1)} \norm{k_{\bar{\j}}}^{-\tau}}\ ,\quad
{\norm{k_{\bar{\j}}}\leq D'\langle\xi_{l_1}\rangle
^{\ep}}\ .
$$
But, if $C_{s}>K'$, $D_s>D'$ and
${\delta_s\geq \delta_{s-1}+\epsilon(d+\tau+1)}$, this means that $\xi$ is also
resonant with $k_{\bar\j } $, and thus it belongs to $Z^{(s+1)}_{M''}$
with $M'':=span_{\Z}(M,k_{\bar\j })$, but this contradicts the fact
that $\xi\in B^{(s)}_M$.
\end{proof}

\begin{lemma} \label{lemma speriamo sia vero}[Separation of
    resonances]   There exist positive constants $\tilde{C}_s$ and $\tilde{D}_s$ depending only on ${\frak c}, {\frak C}, d, \ep, {\tau,} \delta_{s-1}, C_{s-1}, D_{s-1}$ such that, if
	$$
	C_s > \tilde{C}_s\,, \quad D_s >\tilde{D}_s,
	$$
        then the following holds true. Let $\zeta \in E^{(s)}_{M}$ for
        some $M$ of dimension $s= 1, \dots, d-1,$ and let $k^\prime$
        be such that 
	\begin{gather*}
	\norm{k^\prime} \leq \langle \zeta_{k^\prime} \rangle^\ep\,,
	\end{gather*}
	then $\forall M^\prime \not\subset M$ s. t.  $s':=\dim M^\prime {\leq} s$ one
        has 
	$$
	\zeta + k^\prime \notin Z^{(s')}_{M^\prime} \quad \,.
	$$
\end{lemma}
\begin{proof}
	The proof is very similar to that of Lemma \ref{lemma che bei
          blocchi}. Assume by contradiction that $\zeta + k^\prime
        \in Z^{(s')}_{M^\prime}$ for some $M^\prime \neq M.$ It follows
        that there exist $s$ integer vectors, $k_1,...,k_{s'}\in M'$ {\it
          among which at least one does not belong to $M$} s.t.
\begin{equation}
  \label{bei.11}
\left|\scala{\zeta_{k_1}+k'}{k_j}\right|\leq
C_{j-1}\langle\zeta_{k_1}+k'\rangle^{\delta_{j-1}}{\norm{k_j}^{-\tau}}\ ,\quad \norm{k_j}\leq
D_{j-1}\langle\zeta_{k_1}+k'\rangle ^{\epsilon}\ . 
\end{equation}
Let $k_{\bar\j }$ be the vector which does not belong to $M$. By
\eqref{diam.1} there exists $\xi\in B^{(s)}_M$
s.t. $\norm{\xi-\zeta}\leq F\langle\xik\rangle^{\delta_{s-1}+\epsilon
  d}$. Since in particular 
$\xi\in Z_M^{(s)}$ there exist $l_1,...,l_s\in M$, linearly independent, s.t.
\begin{equation}
  \label{bei.41}
\left|\scala{\xi_{l_1}}{l_j}\right|\leq
C_{j-1}\langle\xi_{l_1}\rangle^{\delta_{j-1}}{\norm{l_j}^{-\tau}}\ ,\quad \norm{l_j}\leq
D_{j-1}\langle\xi_{l_1}\rangle ^{\epsilon}\ . 
\end{equation}
We now apply Lemma \ref{xixik} with $h:=k_{\bar\j }/2$, $\ell:=l_1/2$,
$\sigmav:=\zeta+\kappa+k'$, $\eta:=\xi+\kappa$. The only nontrivial
assumption of Lemma \ref{xixik} to verify is the first of
\eqref{xixi.2}. One has
$$
\norm{\xi-\zeta-k'}\leq \norm{\xi-\zeta}+\norm{k'}\leq F\norm{\xik
}^{{\delta_{s-1}+\epsilon
  d} } +\norm{k'}\ . 
$$
To estimate $\norm{k'}$ we proceed as follows:
$$
\norm{k'}\leq D_0\left\langle\zeta+\kappa+\frac{k'}{2}\right\rangle^\epsilon \leq
D_0 K\left(\langle\zeta+\kappa\rangle^\epsilon
+\frac{1}{2^\epsilon}\langle k'\rangle^\epsilon \right) \ ,
$$
where we used eq. \eqref{tri.1}. Using Lemma \ref{xallaa}, we get
$\norm{k'}\leq K''\langle\zeta+\kappa\rangle^\epsilon$ and therefore
$$
\norm{\xi-\zeta-k'}\leq K\norm{\xik
}^{{\delta_{s-1}+\epsilon
  d} }
$$
Thus \eqref{xixi.5}
implies
$$
\left|\scala{\xi_{l_1}}{k_{\bar\j }}\right|\leq
K'{\langle\xi_{l_1}\rangle ^{{\delta_{s-1}+(d+\tau + 1)\epsilon}}\norm{k_{\bar{\j}}}^{-\tau}}\ ,\quad
\norm{l_1}\leq D'\langle\xi_{l_1}\rangle
^\epsilon\ .
$$
But, if $C_{s}>K'$, $D_s>D'$, this means that $\xi$ is also
resonant with $k_{\bar\j } $, and thus it belongs to $Z^{(s+1)}_{M''}$
with $M'':=span_{\Z}(M,k_{\bar\j })$, and this contradicts the fact
that $\xi\in B^{(s)}_M$.
\end{proof}

The following theorem summarizes the result of this subsection

\begin{theorem} \label{cor partizione}
Under the hypotheses of Theorem  \ref{teo dieci piccoli blocchi}, the
blocks $E^{(0)},\ E^{(d)},\ \{E^{(s)}_M\}_{s, M}$ are a partition of
$E$. {Furthermore $E^{(d)}$
	has dimension less then $n_*<\infty$, {with $n_*$ only depending on ${\frak c}, {\frak C}, \delta, \ep, \tau\,$} and, if $E = \Z^d,$  $E^{(0)}$ is of density 1 at infinity.}
\end{theorem}
\begin{proof}
Let $M_1$ and $M_2$ be two submodules of respective dimension $s_1$ and $s_2\,.$ If $s_1 > s_2,$ by definition of the extended blocks one has
$ \displaystyle{E^{(s_1)}_{M_1} \cap E^{(s_2)}_{M_2} = \emptyset\,.}$
Let then $s_1 = s_2\,:$ by Lemma \ref{lemma che bei blocchi},
$$
E^{(s_1)}_{M_1} \cap Z^{(s_2)}_{M_2} = \emptyset\,,
$$
hence, being $\displaystyle{E^{(s_2)}_{M_2} \subseteq Z^{(s_2)}_{M_2}\,,}$ it follows that $\displaystyle{E^{(s_1)}_{M_1}}$ and $\displaystyle{E^{(s_2)}_{M_2}}$ have no intersection.
\end{proof}

\subsection{Invariance of the sets $E^{(s)}_M$.}

Consider now an operator of the form
\begin{gather} 
{\cal L} = \widetilde{H} + {\cal R} \,,\\
\widetilde{H} := -\Delta_{g, \kappa} + {\NForm}\,,\quad
\cR\in\change{\ops {\mu-\tN\rho}}
\label{mult + nform}
\end{gather}
with $\NForm$ in resonant normal form. Since a Fourier multiplier like $-\Delta_{g,\kappa}$,
        leaves invariant any set of the form \eqref{def cal e}, we
        focus on ${\NForm}$ only.

Remark that, in order to study if a set is invariant, we have to analyze
the indices $\csi, \zeta \in E\subset \Z^d$ s.t.
	$$
	\langle \NForm [e^{\ii \csi \cdot x}]\,, e^{\ii \zeta \cdot x} \rangle \neq 0\,.
	$$ 

\begin{lemma}
  \label{supporto}
  Let $\NForm = \Op(\nform),$ $\nform (x, \xi)=\sum_{k\in
    \Z^d}\hat \nform _k(\xi)e^{\ii k\cdot x},$ be a normal form
  operator; let $M$ be a submodule with dim$M\geq 1$, then
  \begin{equation} 
    \label{sopporto}
 \xi\in E^{(s)}_M\quad \Longrightarrow \NForm [e^{\ii {\csi} \cdot x}] =\sum_{k\in
    M}\hat \nform_k \left( \xi + \frac{k}{2}\right)e^{\ii (k + {\csi})\cdot x}\ .
  \end{equation}
\end{lemma}
\proof         By the definition of Weyl quantization one has 
\begin{gather*}
{\NForm}  [e^{\ii \csi \cdot x}] = \sum_{k \in \Z^d} \hat{\nform}_k\left(\csi + \frac{k}{2}\right) e^{\ii (\csi + k) \cdot x}\,.
\end{gather*}
In particular, given $\csi \in \Z^d\,,$  
$$
\langle {\NForm} [e^{\ii \csi \cdot x}]\,, e^{\ii (\csi + k) \cdot x} \rangle  \neq 0
$$
implies that, either $k = 0,$ or 
$$
\left(\csi + \frac{k}{2}\right) \in \textrm{supp}(\hat{\nform}_k)\,.
$$
Assume now by contradiction that $\exists k\not \in M$ s.t. $\hat \nform_k \left(\xi + \frac{k}{2}\right)\not=0$; since $N$ is in normal form this implies in
particular 
$$
\left|\scala{ \xi_k}{k}\right|\leq \langle\xi_k\rangle^\delta\ ,\quad
\norm{k}\leq \langle\xi_k\rangle^\epsilon\ ,
$$
which means that, defining $M':={\rm span}_{\Z}k$, that $\xi\in
Z^{(1)}_{M'}$, with $M'\not\subset M$. This conclusion however is in
contradiction with the conclusion of Lemma \ref{lemma che bei blocchi}.\qed

The main result of this subsection is the following theorem.

\begin{theorem} \label{teo dieci piccoli blocchi} Let $E\subset\Z^{d}$
  and let $\cE\subset L^2(\T^d)$ be the corresponding subset of $L^2$.  
There exists a choice of the constants $C_1\,, \dots, C_{d-1}\,,$
$D_{1}\,, \dots, D_{d-1}$ in Definition \ref{def res zones}
and in Equation \eqref{potenziale risonante} 
such that $\forall s,M$ the set $\cE^{(s)}_{M}$ is left
                  invariant by an operator $\NForm$ in normal form,
                  namely: if $\zeta \in E^{(s)}_{M}$ and $ \langle
                  {\NForm [e^{\ii \zeta \cdot x}], e^{\ii \xi \cdot x
                    }\rangle \neq 0\,,} $ then $\xi \in
                    E^{(s)}_{M}\,.$ Furthermore, in such a
                    case one has
		\begin{equation} \label{piano affine}
		\zeta - \xi \in M\,.
		\end{equation}
{Furthermore, the constants $C_1, \dots, C_{d-1}$ and $D_1, \dots,
  D_{d-1}$ depend on the parameters $d, \ep, \delta, \tau, {\frak c},
  {\frak C}\,$} only.
\end{theorem}
\begin{proof} Take $\zeta\in E^{(s)}_M$, assume that $\xi$ is such that
\begin{equation}
  \label{xizeta}
\langle e^{\ii\xi\cdot x};\NForm [e^{\ii\zeta\cdot x}]\rangle\not =0\ .
\end{equation}
First we remark that, by Lemma \ref{supporto}, one has
$$
\NForm [e^{\ii\zeta\cdot x}]=\sum_{k\in M}\hat
\nform_k\left(\zeta+\frac{k}{2}\right) e^{\ii(\zeta+k)\cdot x}\ ,
$$
so, in particular
$$
\eqref{xizeta}\ \Longrightarrow\ \xi-\zeta\in M
$$
and also
\begin{equation}
  \label{xizeta.1}
\xi=\zeta+k\ ,\quad \norm{k}\leq \langle\zeta_k\rangle^\epsilon\ .
\end{equation}

We now proceed in proving that \eqref{xizeta} also implies $\xi\in
E^{(s)}_M$. 

First, if $M=\{0\}$, then, by the very definition of normal form, $\cN$
acts as a Fourier multiplier on $E^{(0)}$, and thus in particular it is
diagonal and leaves it invariant. Furthermore, $E^{(0)}$ decomposes into
invariant subspaces. Each one of these subspaces is just a single point of $\Z^d={\Mc}$.

In order to prove the result for higher values of $s$, we first remark
that
\begin{equation*}
E^{(s)}_M=\left( \left\{B^{(s)}_M +M \right\}\cap Z^{(s)}_M
\right)\setminus \left(\bigcup_{r<s} E^{(s)}\right)\ .
\end{equation*}

From \eqref{xizeta.1} it follows that $\xi\in E^{(s)}_M+M\subset
B^{(s)}_M+M  $. We are going to prove by induction on $s$ that $\xi\in
Z^{(s)}_M$ and that it also belongs to the complement of $\bigcup_{r<s}
E^{(s)}$.

We know the result is true for $s=0$. By induction we have that if
$\zeta\in E^{(s-1)}_M$ then $\xi\in E^{(s-1)}_M$, and therefore also
$\xi\in Z^{(s-1)}_M$; we prove now that if
$\zeta\in E^{(s)}_M$ then $\xi\in Z^{(s)}_M$. Assume by
contradiction that this is not true. Since the sets $\{E^{(\tilde
  s)}_{\tilde M}\}_{\tilde{s}, \tilde{M}}$ form a partition, then there exists $s'$, and $M'\not=M$
s.t. $\xi\in E^{(s')}_{M'}\subset Z^{(s')}_{M'}$.

There are three cases
\begin{itemize}
\item[1)] $s'=s$. Then, by \eqref{xizeta.1}, one can apply Lemma
  \ref{lemma speriamo sia vero}, which implies
  $$
\xi\not\in Z^{(s)}_{M'}\ , \quad \text{{\rm
    unless}} \ M=M'\ .
$$
Thus this case is not possible.

\item[2)] $s'>s$. By Remark \eqref{rmk zone inscatolate}, and item 1),
  this implies $\xi\in Z^{(s)}_M$, against the contradiction assumption.

\item[3)] $s'<s$. Just remark that \eqref{xizeta} is equal to
\begin{equation}
  \label{xizeta.4}
\langle e^{\ii\xi\cdot x};\NForm [e^{\ii\zeta\cdot x}]\rangle=\langle \NForm
[e^{\ii\xi\cdot x}];e^{\ii\zeta\cdot x}\rangle\not =0\ ,
\end{equation}
but the inductive assumptions says that $E^{(s')}_{M'}$ is
invariant for $s'<s$, thus \eqref{xizeta.4} implies $\zeta\in
E^{(s')}_{M'}$ which is impossible since the extended blocks form a
partition. 
\end{itemize}

Thus we have $\zeta\in E^{(s)}_M$ then $\xi\in \left\{B^{(s)}_M+M
\right\}\cap Z^{(s)}_M$. Then by induction, using \eqref{xizeta.4},
$\xi\in E^{(s')}_{M'}$, $s'<s$, implies $\zeta\in E^{(s')}_{M'}$
and
thus $\zeta\in E^{(s)}_{M}$ implies $\xi\not\in E^{(s')}_{M'}$,
$\forall s'<s$, and this concludes the proof. 
\end{proof}

By equation \eqref{piano affine}, each extended block is foliated in
equivalence classes left invariant by an operator in normal form. {\it
  We
define the sets $W_{M,\beta}$ of Theorem \ref{main} to be such
equivalence classes}. We are now going to show that they are labeled
by $\beta$ in a subset of ${\Mc}$. First remark that, if $\xi\in E^{(s)}_M$, there
exists $W_{M,\beta}$ s.t. $\xi\in W_{M,\beta}$ and then one has
$$
W_{M,\beta}\subset \xi+M\ . 
$$
Introduce now a basis adapted to $M$, then, since
$\Z^d=M+{\Mc}$, for any equivalence class there exists $\beta\in {\Mc}$
s.t. $W_{M,\beta}\subset \beta+M$. Conversely, given $\beta\in  {\Mc}$ we
define
$$
W_{M,\beta}:= (\beta+M)\cap E^{(s)}_M\ ,
$$
which is possibly empty. Following Definition \ref{tildee},
$\widetilde M$ is the subset of the $\beta$'s
 s.t. $W_{M,\beta}$ is not empty.

We have thus established the following Corollary.

\begin{corollary}\label{cor invarianza}
	\change{Given $\ep, \delta \in \R^+$ and
	$\tau > d - 1$ fulfilling \eqref{legami parametri},
	a
	Floquet parameter $\kappa$ and a flat metric $g$, there exists a partition $\{ W_{M, \beta} \}_{ M \subseteq \Z^d, \beta \in  {\Mc}}$  of $\Z^d$
	which is left invariant by any operator in normal form. Furthermore,  all the sets $W_{M, \beta}$ have finite cardinality, and they satisfy the following:
	\begin{itemize}
	\item[(i)] the set $E_{\{0\}}:=\bigcup_{\beta}W_{\{0\}, \beta}$
	has density one at infinity and $\sharp W_{\{0\}, \beta} = 1$ for any $\beta \in \Z^d$.
	\item[(ii)] If $M = \Z^d$, $\Mc = \{0\}$ and $W_{\Z^d,\{0\}}$ has cardinality bounded
	by an integer $n_*$ 
	which depends on the constants of the metric and on $d,\ \delta,\ \ep,
	{\tau}$ only
	\item[(iii)]
	If $M$ is a proper module, for any $\beta \in {\Mc}$ one has: $\csi^{\bot}_M = (\csi')^\bot_M$  $\forall \csi, \csi' \in W_{M, \beta}$, and
	\begin{equation}
	\| \csi_M\| \leq K \langle \csi \rangle^{\delta + d(d + \tau + 1)\ep} \quad \forall \csi \in W_{M,\beta}\,,
	\end{equation}
	where $K$ is a positive constant depending only on the constants of the metric and on $d,\ \delta,\ \ep,
	{\tau}.$
	\end{itemize}
	}
\end{corollary}

\subsection{Dimensional reduction}\label{rid}

We analyze now the restriction of $\widetilde H$ to each invariant
set. \change{We point out that, as observed in Remark \ref{rmk m}, this is the only point of the whole Section \ref{sez blocchi} where the boundedness condition $\mu\leq 0$ is needed. We then} consider 
\begin{equation} \label{decomp h}
\begin{gathered}
\widetilde{H}_{M, \beta}\equiv \Pi_{{\cal W}_{M, \beta}} \left(-\Delta_{g, \kappa} + \NForm_M \right)  \Pi_{{\cal W}_{M, \beta}}\,, 
\end{gathered}
\end{equation}
with 
\begin{equation} \label{vm}
\NForm_M = \Op(\nform_M)\,, \quad \nform_M(x, \csi) = \sum_{ k \in M} \hat{\nform}_k (\csi) e^{\ii k \cdot x}\,,
\end{equation}
in normal form.

{Given $\xi\in W_{M,\beta}$, let $\tilde{\csi}$ and $\kappa^\prime$ be defined as in \eqref{betainz}, namely
	$$
	\tilde{\csi} = \csi - \inte{(\xi+\kappa)_M}\,, \quad \kappa'=\left\{ (\xi+\kappa)_M   \right\}\,, 
	$$ and recall that, as pointed out in Remark \ref{rmk costanti}, one has $\tilde{\csi} = \tilde{\beta}\,.$ Thus, defining
\begin{equation}
  \label{dec.3}
\zeta:=\inte{(\xi+\kappa)_M}\ , \quad \ell^2:=\norma{(\tilde
  \beta+\kappa)_{M^{\perp}}}^2\ , 
\end{equation}
one has
\begin{align}
  \label{intern}
  \xi=\zeta+\tilde \beta\ ,\quad (\xi+\kappa)_M=\zeta+\kappa'\ ,
  \\
  (\xi+\kappa)_{M^\perp}=  (\tilde \beta+\kappa)_{M^\perp}\ ,
  \\
    \label{intern.2}
  \norm{\xi+\kappa}^2=\norm{\zeta+\kappa'}^2+\ell^2\ .
\end{align}
}
\begin{remark}
  \label{trasla}
Consider the translation $W_{M,\beta}\ni\xi\mapsto \zeta=\xi-\tilde
\beta\in W^t_{M,\beta}\subset M$; {as pointed out in Remark \ref{Gauge},} its quantization is the
Gauge transformation $U_{\beta}=e^{-i\tilde \beta\cdot x }$. By
standard pseudodifferential calculus, given a symbol 
$a(x,\xi)$ one has that the symbol of ${U_{\beta}
\Op(a)U_{\beta}^{-1}}$ is
\begin{equation}
  \label{trasl.1}
a^{trasl}(x,\zeta):=a(x,\zeta+\tilde \beta)\ ,
\end{equation}
which, if $a$ is in normal form, is a function on
$T^*\T^{s}$.
\end{remark}

Precisely, we have the following lemma
\begin{lemma} \label{prop riduzione}
	With the above notations, assume that $\nform_M\in\simb m $
        with $m\leq 0$, is in normal form with respect to $M$, then,
        {in coordinates adapted to $M$}, one has
        \begin{equation}
          \label{rida}
{U_{ \beta}}  \left.(-\Delta_{g, \kappa}+\NForm_M)\right|_{\cW_{M,\beta}}
{U_{ \beta}^{-1}}= {\left.\left(-\Delta_{g, \kappa'}+\NForm'_M+\ell^2\right)\right|_{\cW^t_{M,\beta}}}\ ,
        \end{equation}
        where $-\Delta_{g, \kappa'}$ is the Laplacian (in $s$
        dimensions) with respect to the restriction of the metric to
        $M$ and
        $$
N'_{M}(x,\zeta)=N_M(x,\zeta+\tilde \beta)\ 
$$
is of class $\simb m$ (as a symbol on $\T^s$), with seminorms bounded
by the seminorms of $\NForm_M$. 
\end{lemma} 
\begin{proof} First remark that, by \eqref{intern.2} the
  transformation of the Laplacian is $-\Delta_{g, \kappa'}+\ell^2$.

We come to the transformation of $\NForm_M$. We observe that, since
it is in normal form with respect to $M$ its symbol has the structure
$$
\nform_M(x,\xi)=\sum_{k\in M}\hat{\nform}_k(\xi)e^{\im k\cdot x}\ .
$$
Furthermore, introducing a basis $\bv^A$ adapted to $M$, and denoting
by $\bu_A$ its dual basis, one has, for $k\in M$,
$$
k\cdot x=\sum_{a=1}^{d'}x^ak_a 
$$ (since the coordinates $k_A$, $A=d'+1,...,d$ of a vector in $M$
vanish). Thus one gets that the symbol $N'_M$ of the transformed
operator is
$$
N_M'(\zeta,\hat z)=\sum_{k\in\Z^{d'}}\hat{\nform}_{k_{a}\bv^a}(\zeta+ \perpz
\beta)e^{\im x^ak_a}=\nform_M (\zeta'+\perpz\beta,\hat x)\ ,\quad \hat x:=(x^1,...,x^{d'})
$$

Remark that, denoting $M_R:= {\rm
  span}_{R}(\bv_1,...,\bv_{d'})$ and $M_{R}^*:= {\rm
  span}_{R}(\bu_1,...,\bu_{d'})$, one has 
\begin{align*}
\norma{d^{N_2}_{\hat{x}} d^{N_1}_{\zeta'} \nform^\prime_M(\hat{x}, \zeta')
}=\sup_{\begin{subarray}\norma{\|h^{(j)}}\|=1,\ h^{(j)}\in
    M^*_R\\ \norma{k^{(j)}}=1,\  k^{(j)}\in M_R 
  \end{subarray}}
\left|d^{N_2}_{\hat{x}} d^{N_1}_{\zeta'} \nform_M^\prime(\zeta',\hat{z}) \left[h^{(1)}\,, \cdots, h^{(M)}, k^{(1)}\,, \cdots, k^{(N)}\right]\right|\\
\leq
\sup_{\begin{subarray}\norma{\|h^{(j)}\|}=1,\ h^{(j)}\in
    \R^d\\ \norma{k^{(j)}}=1,\  k^{(j)}\in \R^d 
  \end{subarray}}
\left|d^{N_2}_{{x}} d^{N_1}_{\xi} \nform_M (\zeta'+\perpz \beta,\hat{z}) \left[h^{(1)}\,, \cdots, h^{(M)}, k^{(1)}\,, \cdots, k^{(N)}\right]\right|\\
=
\norma{d^{N_2}_{{x}} d^{N_1}_{\xi} \nform_M (\hat {x},\zeta'+\perpz\beta )
}\leq C\langle \zeta'+\perpz\beta+\kappa  \rangle^{m-N_1\delta}\leq
C\langle (\zeta'+\perpz\beta+\kappa)_M  \rangle^{m-N_1\delta}
\\
=C\langle \zeta'+\kappa' \rangle^{m-N_1\delta}\ .
\end{align*}
which is the thesis.
\end{proof}
In order to deduce the spectral result, the following corollary will
be useful
\begin{corollary}
  \label{speca}
  Let $\norma{\zeta+\kappa'}_{}^2+m(\zeta)$ be an eigenvalue of
  $\left.(-\Delta_{g, \kappa'}+\NForm'_M)\right|_{\cW_{M,\beta}^t}$
  with eigenfunction $\phi^{(\zeta)}.$ Then
  $\norma{\xi+\kappa}_{}^2+m(\xi-\tilde\beta) $ is an eigenvalue of 
$\left.(-\Delta_{g,
    \kappa}+\NForm_M)\right|_{\cW_{M,\beta}}$ with eigenfunction
  $\psi^{(\xi)}:=e^{\im \tilde \beta\cdot x}\phi^{(\zeta)}$.  
\end{corollary}
\begin{remark}
  \label{mulo}
By \eqref{intern}, in the particular case where $\phi^{(\zeta)}=e^{\im
  \zeta\cdot x}$, one has $\psi^{(\xi)}=e^{\im \xi\cdot x}$.
\end{remark}

\section{A spectral result by quasi-modes}\label{quasiquasi}

In this section we prove Theorem \ref{speci_spettri}.

The key quasimode argument we are going to use is a variant of that
used in \cite{BKP15} (see Proposition 5.1) and is the following one

\begin{lemma}[Quasi-mode argument] \label{prop quasimodi}
	Let $H = H_0 + H_1$ be a self-adjoint operator on the Hilbert
        space ${\cal H}$ such that $H$ and $H_0$ have pure point
        spectrum. Suppose that $\lambda^{(0)}_{1} \leq \dots \leq
        \lambda^{(0)}_{M}$ are $M$ eigenvalues of $H_0$ counted with
        multiplicity, such that $\exists D>0$ with
	\begin{equation} \label{ho dei gaps}
	(\lambda^{(0)}_1 - D\,,\ \lambda^{(0)}_1)\,, \quad (\lambda^{(0)}_M\,,\  \lambda^{(0)}_M + D)
	\end{equation}
which contain no eigenvalues of $H_0$. Denote by $\{ \psi_{k}
\}_{k=1}^{M}$ the orthonormal eigenfunctions corresponding to $\{\lambda^{(0)}_k\}_{k=1}^{M}$, and
let $\{\varepsilon_k\}_{k = 1}^M $ be such that
	\begin{equation}\label{ipotesi H1 quasi modi}
	\begin{gathered}
	\| H_1 \psi_k \| \leq \varepsilon_k \ ,\quad k = 1\,, \dots\,, M\,.
	\end{gathered}
	\end{equation}
If $D > 0$ and $\delta \in (0, 1)$ are such that 
	\begin{equation}\label{ho dei grossi gaps}
	\begin{gathered}
	D^2 \geq \frac{16}{\pi \delta^2} M^3 \left( \max_{k} \varepsilon_k \right) \left(|\lambda^{(0)}_{M} - \lambda^{(0)}_{1}| + D\right)\,,
	\end{gathered}	
	\end{equation}
then there are at least $M$ (not necessarily distinct) eigenvalues of
$H$ in the interval
	$$
	(\lambda^{(0)}_{1} - \delta D\,,\ \lambda^{(0)}_{M} + \delta D)\,. 
	$$
\end{lemma}
\begin{proof}
	By contradiction, assume that there are less than $M$
        eigenvalues of $H$ inside the interval $\left(\lambda^{(0)}_1
        - \delta D, \lambda^{(0)}_M + \delta D\right)$, with $\delta \in (0, 1)$. In
        particular, there are less than $M$ eigenvalues in the
        intervals
	$$
	I^- = \left(\lambda^{(0)}_1 - \delta D,\ \lambda^{(0)}_1\right)\,, \quad I^+ = \left(\lambda^{(0)}_M,\ \lambda^{(0)}_M + \delta D \right)\,.
        $$ Since $I^{+}$ has length $\delta D$, there exists at
        least one interval $J^+ \subset I^+$ such that $|J^+| \geq \frac{\delta D}{M}$ which contains no eigenvalues of $H$,
        analogously for $I^-\,:$ there exists at least an interval
        $J^- \subset I^-$ containing no eigenvalues of $H$ and having
        length $|J^-| \geq \frac{\delta D}{M}$. Remark that, by hypothesis
        \eqref{ho dei gaps}, $J^\pm$ do not contain eigenvalues of
        $H_0$ either.  Consider then a square closed path $\gamma$ in
        the complex plane intersecting the real axis at the middle
        points of $J^+$ and $J^-\,.$ By construction,
	\begin{equation}\label{lontano da h0}
	\textrm{dist}\left(\gamma, \sigma(H_0)\right) \geq \frac{\delta D}{2 M}\,,
	\end{equation}
	and
	\begin{equation}\label{lontano da h}
	\textrm{dist} \left(\gamma, \sigma(H)\right) \geq \frac{\delta D}{2 M}\,.
	\end{equation}
	Moreover, the length $\ell(\gamma)$ of $\gamma,$ fulfills
	\begin{equation}\label{its a long way}
	\ell(\gamma) \leq 4 \left|\lambda^{(0)}_M - \lambda^{(0)}_1 +
        D \right|\,. 
	\end{equation}
By the contradiction assumption there are less than
$M$ eigenvalues of $H$ inside the path $\gamma\,.$ Let $M_0$ be their
number (counted with multiplicity).

Denote by 
	$R(z)= \left( H - z \mathbb{I}\right)^{-1}$ the resolvent of
$H$, and by $R_0$ the resolvent of $H_0$, 
then, if $P$ denotes the
projection operator on the eigenspace corresponding to
such eigenvalues of $H$, one has
	$$
	P = \frac{1}{2\pi i }\int_{\gamma} R(z)\ dz\,,
	$$
and, using the resolvent identity
	$$
	R(z) - R_0(z) = R(z) H_1 R_0(z)\,,
	$$
one has
        \begin{align}
\label{B.7a}	P \psi_k = \frac{1}{2\pi i } \int_{\gamma} R_0(z)\ dz \psi_k + \frac{1}{2\pi i } \int_{\gamma} R(z) H_1 R_0(z)\ dz\ \psi_k = \psi_k + r_k\,.
	\end{align}
	where
	$$
	r_k = \frac{1}{2\pi i } \int_{\gamma} R(z) H_1 R_0(z)\ dz\ \psi_k\,.
	$$
	By  \eqref{lontano da h0} and \eqref{lontano da h}, using that $R_0(z) \psi_k = \frac{1}{\lambda_k^{(0)} -z} \psi_k$, and the hypothesis \eqref{ipotesi H1 quasi modi}, one gets the estimate 
	one has that
	\begin{align}\label{rk piccolo}
	\| r_k\| \leq \frac{\ell(\gamma)}{2\pi} \Big(\frac{2 M}{\delta D} \Big)^2 \varepsilon_k \,.
	\end{align}
We are going to show that the vectors \eqref{B.7a} are independent,
against the assumption $M_0<M$.  We prove that 
	$$
	\sum_{k=1}^{M} \alpha_k P \psi_{k} = 0\,
	$$
implies $\alpha_k=0$, $\forall k$.	Indeed, one has
	$$
	\sum_{k=1}^{M} \alpha_k P \psi_{k} = \sum_{k=1}^{M} \alpha_k  \left(\psi_{k}  + r_k \right) =  0\,;
	$$
	in particular,
	$$
	\sum_{k=1}^{M} \alpha_k \left(\langle \psi_k\,,\ \psi_j \rangle + \langle r_k\,,\ \psi_j \rangle\right) = 0 \quad \forall\ j\,,
	$$
	namely
	$
	\left(\mathbb{I} + A\right) \alpha = 0\,,
	$
	with $A$ the $M$ dimensional matrix with matrix elements given by $A_{k,j} = \langle \psi_k, r_j \rangle\,.$
	Since by \eqref{rk piccolo}
	$$
	\| A \| \leq M \sup_{i, j} \{|A_{i,j}| \} \leq M
        \frac{\ell(\gamma)}{2\pi} \Big(\frac{2 M}{\delta D} \Big)^2 \varepsilon\ , \quad \varepsilon := {\max_{k}}\varepsilon_k\,. 
	$$
        Then hypothesis \eqref{ho dei grossi gaps} ensures $\|A\|<
        1\,,$ so that $ \mathbb{I} + A $ is invertible and thus $\alpha_k = 0 \ \forall k\,.$ This shows that
        $\{P\psi_k\}_{k =1}^{M}$ form a set of $M$ linearly
        independent eigenfunctions, which contradicts the hypothesis
        that there is only a set of multiplicity $M_0 < M$ of
        eigenvalues of $H$ inside the interval $\left(\lambda^{(0)}_1
        - \delta D\,, \lambda^{(0)}_M + \delta D\right).$
\end{proof}

The main tool in order to describe the unperturbed spectrum is a Weyl
type estimate for the eigenvalues of an operator $H^{(0)}$ with
spectrum given by
\begin{equation} \label{spettro h0 e}
\begin{gathered}
\sigma(H_0) = \left \lbrace h_0(\csi)\ |\ \csi \in E\subset \Z^d
\right \rbrace\,, \\ h_0(\csi) = \|\csi + \kappa\|^2 +
m(\csi) \quad \forall \csi \in E\,,
\end{gathered}
\end{equation}
with $ m$ a bounded function.

\begin{lemma} \label{weyll}
Consider an
operator $H_0$ as above and denote $\tM = \sup_{\xi \in E} | m
  (\xi)|$, let $R>\sqrt{3\tM}$, then one has
  \begin{equation}
  \label{weyl_abs}
\#\{ \xi\ :\ |h_0(\xi)|\leq R^2 \}\leq \left(\frac{4}{{\frak
    c}_1}\right)^dR^d\ . 
\end{equation}
\end{lemma}
\proof An estimate of the quantity \eqref{weyl_abs} is the number of
points $\xi\in\Z^d$ contained in a ball centered at $-\kappa$ and having
 radius $\sqrt{R^2+\tM}\leq 2
R$. Of course the ball is defined in terms of the metric
$g^*$. For any $\xi\in\Z^d$, consider a ball $B_{{\frak
    c}/2}(\xi)$ of radius ${\frak
    c}/2$ and center $\xi$. Then, as $\xi$ varies, such balls do not
intersect, thus the ``volume occupied'' by $n$ points of the lattice
is bigger than $n$Vol$B_{{\frak
    c}/2}(\xi)=n C_d({\frak c}/2)^d, $ with $C_d$ the volume of the
unitary ball. It follows that for the number $n$ of points in the ball of
radius $2R$ (independently of its center) the following inequality holds:
$$
n C_d({\frak c}/2)^d \leq {\rm Vol}B_{0}(2R)=C_d2^dR^d\ ,
$$
from which the thesis follows. \qed

This allows to prove the existence of gaps in the spectrum; precisely,
the following Lemma holds.

\begin{lemma}
  \label{gapss}
There exists a constant $C$, depending only on ${\frak c}$ and $d$,
with the following properties: for any $\bar\lambda >4\tM$ and any
$0<L\leq \tM$, there exist $0<L_1,L_2<L$ s.t.
\begin{align}
  \label{buchi}
\# \left(\sigma(H_0)\cap [\bar \lambda-L_1;\bar\lambda+L_2]\right)\leq
C \bar \lambda^{d/2}
\\
\label{buchiveri}
\sigma(H_0)\cap \left[\bar \lambda-L_1-\frac{L}{C\bar
  \lambda^{d/2}};\bar\lambda-L_1\right]=\emptyset
  \\
\label{buchiveri2}
\sigma(H_0)\cap \left[\bar \lambda+L_2;\bar\lambda+L_2+\frac{L}{C\bar
  \lambda^{d/2}}\right]=\emptyset
  \end{align}
\end{lemma}
\proof By Lemma \ref{weyll} the maximal number of eigenvalues smaller
than $\bar\lambda+L<2\bar\lambda$ is smaller than a constant $C$ (the
constant whose existence is claimed in the statement) times
$\bar\lambda^{d/2}$, so equation \eqref{buchi} is true (and very
pessimistic) for any choice of $L_1, L_2 < L.$ To prove \eqref{buchiveri}, consider the interval
$[\bar\lambda-L,\bar\lambda]$; by Lemma \ref{weyll} it contains at most
$C\bar\lambda^{d/2}$ eigenvalues, so there is at least a gap between
two of them of length $L/C\bar\lambda^{d/2}$. Its right end determines
$L_1$, and this proves \eqref{buchiveri}. Equation \eqref{buchiveri2}
is proved in the same way. \qed

\begin{corollary}
  \label{spettro}
  For any $N>0$ and $0<L<\tM$, there exists a sequence of intervals
  \begin{equation}
    \label{ej}
E_j=[a_j,b_j]\ ,\quad j \in \N
  \end{equation}
and a positive constant $C$,  with the following properties:
  \begin{align}
    \label{spettri_vari}
\sigma(H_0)\subset\left[0,a_1-\frac{1}{a_1^N}\right]\bigcup\left( \bigcup_{j}E_j\right)\ ,
\\
\label{spettro.1}
\left|b_j-a_j\right|\equiv \left|E_j\right|\leq {2} L
\\
\label{spettro.11}
d(E_j,E_{j+1})\equiv a_{j+1}-b_j\geq\frac{L}{b_j^{N}}
\\
\label{ancora spettri}
\# \left(\sigma(H_0)\cap E_j\right)\leq{C} b_j^{d/2}\ .
  \end{align}
\end{corollary}
\proof We use the same notations as in Lemma \ref{gapss}. Take $\bar\lambda := {\rm min } \{ \lambda \in \sigma(H_0) : \lambda \geq 4 \tM \}$. Then the first interval is the one constructed in Lemma
\ref{gapss}. Let $b_1$ be the largest point of the spectrum in the
interval. Let $a_2$ be the subsequent point of the spectrum. To
determine $b_2$, consider the subsequent points of the spectrum. By
Lemma \ref{gapss} after at most an interval of length ${2}L$ one finds a
gap of width $\frac{L}{a_2^{N}}$. This gives the second
interval. Iterating one gets the result.
\qed

The following Lemma enables to relate the spectrum and the structure of eigenfunctions of the two operators $H^{(1)}_{M, \beta}$ and $\widetilde{H}_{M, \beta}$ of Theorem \ref{main}, for any $M \subset \Z^d$ and $\beta \in \widetilde{M}\,:$
\begin{lemma}\label{in su}
	For any $M,\beta$, consider the operator
	$ - \Delta_{g, \kappa^\prime} + \cV_{M, \beta}$  as in \eqref{onestep} of Theorem \ref{main}, and
	assume that its eigenvalues are given by
	\begin{equation}
	\label{lambdaz}
	\lambda_\zeta=h_{M,\beta}(\zeta)=\norma{\zeta+\kappa'}_{}^2 +m_{M,\beta}(\zeta)
	\,, \quad \zeta\in M,
	\end{equation}
	with $\sup_{M, \beta} \sup_{\zeta}|m_{M,\beta}(\zeta)|\leq
        \tM$. Assume that there exist positive constants $ a <
        \frac{1}{2},$ $\tN \in \N$ and $C $ such that, given any
        eigenvalue  $\lambda_\zeta \neq 0,$ the corresponding
	eigenfunction $\phi^{(\zeta)}$ fulfills 
	\begin{equation}
	\label{lefi}
	\norma{\phi^{(\zeta)}}_{H^{-\tN}}\leq \frac{C}{\lambda_\zeta^{a\tN}} \quad \forall \zeta \in M\,.
	\end{equation}
	Then the eigenvalues of ${U_{\beta}^*}\left(-\Delta_{g, \kappa^\prime} + {\cal V}_{M, \beta}\right) {U_\beta}+ \ell^2 $ are given by
	\begin{equation}
	\label{eima}
	\lambda_\xi=h_0(\zeta)=\norma{\xi+\kappa}_{}^2+m_{M,\beta}(\xi-\tilde\beta)\ ,\quad
	\xi=\zeta+\tilde\beta\,
	\end{equation}
	and, 
	if $\lambda_\csi \neq 0,$ there exists $C^\prime >0,$ depending only on $a, {\tt m}, \tN, C,$ such that the corresponding eigenfunction $\psi^{(\xi)}$ fulfills
	\begin{equation}
	\label{nemenemt}
	\norma{\psi^{(\xi)}}_{H^{-2\tN}}\leq \frac{C^\prime}{\lambda_\xi^{a\tN}}\ .
	\end{equation}
\end{lemma}
\proof The form of the eigenvalues is a direct consequence of
eq. \eqref{intern.2}. Concerning the eigenfunctions, the unitary
map ${U_\beta^*}$ transforms them in $\psi^{(\xi)}:= e^{\im
	\tilde\beta\cdot x}\phi^{(\zeta)}$, which, by Lemma \ref{negativo},
are estimated by  
\begin{equation}
\label{nemenemt.1}
\norma{\psi^{(\xi)}}_{H^{-2\tN}}\leq \frac{C}{\lambda_\zeta^{a\tN}}
\frac{1}{\langle (\tilde\beta+\kappa)_{M^{\perp}}\rangle^\tN}\ .
\end{equation}
Then one has
\begin{align*}
\lambda^{a}_\zeta\langle (\tilde\beta+\kappa)_{M^{\perp}}\rangle\geq
\left(\lambda^{1/2}_\zeta\langle (\tilde\beta+\kappa)_{M^{\perp}}\rangle\right)^{2a}\ ,
\end{align*}
since $2a<1$. Then, provided $\lambda_\zeta$ is large enough,
$\lambda^{1/2}_\zeta\geq \langle(\zeta+\kappa')\rangle/2$, from which
\begin{equation}
\label{}
\lambda^{1/2}_\zeta\langle (\tilde\beta+\kappa)_{M^{\perp}}\rangle\geq
\frac{1}2
\langle\zeta+\kappa'\rangle\langle(\tilde\beta+\kappa)_{M^{\perp}}
\rangle = \frac{1}2
\langle(\xi+\kappa)_{M}\rangle\langle(\xi+\kappa)_{M^{\perp}}
\rangle\geq\frac{1}{2} \langle\xi+\kappa\rangle\ , 
\end{equation}
where the last inequality follows from the trivial remark that for any
real $x,$ $y$, one has $(1+x^2)(1+y^2)\geq1+x^2+y^2$. Collecting the
results and remarking that, for $\lambda_\xi$ large enough,
$\lambda_\xi<2\langle\xi+\kappa\rangle^2$, one gets the thesis for large
eigenvalues. {In order to cover all the non-vanishing eigenvalues, just
remark that the number of eigenvalues smaller than any threshold is
finite, so that the claimed estimates trivially hold.} \qed

\begin{lemma}
	\label{scalata finale}
	Assume that all the operators \eqref{onestep} fulfill the assumptions
	of Lemma \ref{in su}, then the properties \eqref{eima} and
	\eqref{nemenemt} hold,  also for the eigenvalues and the eigenfunctions
	of the operator \eqref{g}, but with new constants depending only on
	the seminorms of $\cV$ and on the constants of the metric, and with a new function $m^\prime_{M, \beta}$ such that
	$$
	m^\prime_{M, \beta}(\csi) = m_{M, \beta}(\csi) + r_\csi\,, \quad |r_\csi| \leq C \| \csi + \kappa \|_{}^{-{a \tN}} \quad \forall \csi\,.
	$$
\end{lemma}
\proof First, by Theorem \ref{main}, for any $\tN^\prime \in \N$, the operator
$-\Delta_{g,\kappa}+\cV$ is unitarily equivalent, through a
pseudodifferential operator $U$ of order 0, to $\widetilde H_{\tN^\prime}+{\cal R}_{\tN^\prime},$ \change{with ${\cal R}_{\tN^\prime} \in \ops{-2\delta \tN^\prime}\,.$} Fix $\tN \in \N,$ let $\tN^\prime = \frac{\tN}{2}\,$ and from now on drop the dependence on $\tN^\prime$ by the operators $\widetilde{H}_{\tN^\prime}, {\cal R}_{\tN^\prime}\,.$ By
Lemma \ref{in su} the eigenvalues of $\widetilde H$ fulfill
\eqref{eima} and \eqref{nemenemt} with $2 \tN$ replaced by $\tN,$ due to the choice of $\tN^\prime\,.$  Concerning the eigenfunctions, we observe that, by \eqref{lefi},  Lemma \ref{prop struttura autof} of the Appendix ensures that there exists a constant $ C^{\prime \prime}>0$ such that any eigenvalue $\lambda_\csi$ of $\widetilde{H} + {\cal R}$ with $\lambda_\csi \neq 0$ has a related normalized eigenfunction $\psi_\csi$ satisfying
\begin{equation}\label{mininorme}
\| \psi_\xi\|_{H^{- \tN}} \leq {C}^{\prime \prime} |\lambda_\csi|^{\frac{d}{2} - a\frac{ \tN}{2}}\,,
\end{equation}
thus \eqref{nemenemt} still holds for the eigenfunctions of $\widetilde{H} + {\cal R}.$ It remains to prove \eqref{eima}. We split $\sigma(\widetilde H)$
according to Corollary \ref{spettro}, choosing $L = 1$ and $\tN = \tN/3$ and in each of the intervals
$E_j$ we apply Lemma \ref{prop quasimodi}. To this end, remark that,
for all eigenvalues $\lambda\in E_j$ one has
$\lambda/2<a_j<b_j<2\lambda$. Let $\phi$ be the eigenfunction of $\widetilde{H}$
corresponding to $\lambda$: then by the Calderon Vaillancourt
Theorem and since the eigenfunctions of $\widetilde{H}$ satisfy eq. \eqref{nemenemt}, one has
$$
{\norma{{\cal R}\phi}_{L^2}} \leq \| {\cal R}\|_{{\cal B}(H^{-\tN}, H^0)} \frac{2^\tN C^\prime}{\lambda^{a \frac{\tN}{2}}}\ .
$$

Thus an application of Lemma \ref{prop quasimodi} with 
$H_0 = \widetilde{H},$  $H_1 = {\cal R},$ 
ensures that, if for all $j \in \N$ one defines $D_j^{-} = a_j^{-\tN/3}, D^+_{j} = b_j^{-N/3}$ and $M_j = \sharp \left( \sigma (\widetilde{H}) \cap E_j\right),$ then there are $M^\prime_j \geq M_j$ eigenvalues of $\widetilde{H} + {\cal R}$ inside the interval
$$
\widetilde{E}_j = \left[a_j - \frac{1}{4} D_j^-\,,\ b_j + \frac{1}{4} D_j^+\right] \supset E_j\,.
$$
We prove now that there are no  eigenvalues of
$\widetilde{H} + {\cal R}$ outside the intervals $\widetilde{E}_j.$
Assume by contradiction
that $\bar{\lambda}$ is an eigenvalue of $\widetilde{H} + {\cal R}$
with $\bar{\lambda} \notin \bigcup_j \widetilde{E}_j$. Let
$\bar{j}$ be the positive integer such that $b_{\bar{j}} <
\bar{\lambda} < a_{\bar{j} + 1}\,.$ Since the eigenfunction $\psi$ of
$\widetilde{H} + {\cal R}$ related to $\bar{\lambda}$ satisfies
\eqref{mininorme}, one has
$$
{\| {\cal R} \psi \|}_{L^2} \lesssim {\bar{\lambda}}^{\frac{d}{2} - a \frac{\tN}{2}} \lesssim {a_{\bar{j} + 1}}^{\frac{d}{2} - a \frac{\tN}{2}}\,,
$$
which implies that $\psi$ is a quasi-mode for $\widetilde{H}$ with
approximated eigenvalue $\bar{\lambda}\,.$ In particular (up to
choosing $a_1$ big enough), this implies that there exists an exact
eigenvalue $\lambda = \bar \lambda + O({a_{\bar{j} + 1}}^{\frac{d}{2} - a \frac{\tN}{2}})$ of $\widetilde{H}$ such that $b_{\bar{j}} <
{\lambda} < a_{\bar{j} + 1}\,,$ which is absurd, by definition of the
intervals $E_j.$

We prove now that $M^\prime_j = M_j$ for all $j \in \N.$ Arguing as
before, we can apply the quasi-mode argument of Lemma \ref{prop
  quasimodi} with $H_0 = \widetilde{H} + {\cal R}$, $H_1 =
- {\cal R}$, $M = M^\prime_j$ and $[\lambda^{(0)}_1,
  \lambda^{(0)}_{M^\prime_j}] = \widetilde{E}_j$ to deduce that, since
all the eigenfunctions of $\widetilde{H} + {\cal R}$ related to the
eigenvalues contained inside $\widetilde{E}_j$ satisfy
\eqref{mininorme}, then there are $M^{\prime \prime }_j \geq
M^\prime_j$ eigenvalues of $\widetilde{H}$ inside a slight enlargement
of the interval $\widetilde{E}_j\,.$ But there are exactly $M_j$
eigenvalues of $\widetilde{H}$ inside $E_j \subset \widetilde{E}_j,$
thus $M^{\prime \prime}_j = M_j\,,$ which proves that \emph{all} the
eigenvalues of $\widetilde{H} + {\cal R}$ are of the form
\eqref{eima}.
We finally observe that, since for any eigenvalue $\lambda^\prime$  of $\widetilde{H} + {\cal R}$ the
corresponding eigenfunction $\psi$ fulfills again equation \eqref{nemenemt}
with updated constants, the corresponding eigenfunction $U \psi$ of $-\Delta_{g, \kappa} + {\cal V}$ fulfills again equation \eqref{nemenemt}, due to the fact that $U$ is a bounded operator onto $H^{-\tN},$ since $U$ is a pseudodifferential operator of order $0.$ \qed

By iteratively applying this Lemma one gets the proof of Theorem \ref{speci_spettri}.

\begin{remark}
  \label{tutti}
From the above Lemma it follows in particular that all the eigenvalues
and eigenfunctions of $\widetilde H+\cR$ are constructed through our
quasimode procedure.
\end{remark}

\appendix

\section{Pseudodifferential calculus}\label{pdc}
In this section we recall some standard facts on pseudodifferential
calculus with the aim of pointing out that, with our coordinate
independent definition of the seminorms, they still hold. In
particular the coordinate independent definition is needed in order to perform the dimensional reduction of
Subsect. \ref{rid}.

\begin{lemma}[\bf Calderon Vaillancourt]
	Let $A \in \OPS^{m, \delta}$. Then $A$ is a bounded linear operator $H^s \to H^{s - m}$ for any $s \in \R$. 
	{In particular, for any $s$ there exist $K>0$ and $N \in \N,$ depending only on the parameters $m, s, d, {\frak c},$ such that $\displaystyle{\|A\|_{ {\cal B} \left( H^{s}\,; H^{s-m}\right)} \leq K \sup_{N^\prime \leq N} C_{N^\prime, 0}(a)\,.}$}
\end{lemma}

Since $T^*\T^d$ is a cotangent bundle it carries a natural
symplectic structure, and the Poisson Brackets can be computed in any
system of coordinates originated by a system of coordinates in
$\T^d$. Using such a system of coordinates, one can easily show that
the the following Lemma holds

\begin{lemma}\label{a.lem pois bra}
	Let $a \in S^{m, \delta}$ and $b \in S^{m^\prime, \delta}\,;$
	then $\{ a, b\} \in S^{m + m^\prime - \delta, \delta}\,.$
	In particular, for all $N_1, N_2 \in \N$ one has
	$$
	C_{N_1, N_2}\left(\{a, b\}\right) \leq  C_{N_1 + 1, N_2} (a) C_{N_1, N_2 +1}(b)+C_{N_1, N_2+1} (a) C_{N_1+1, N_2}(b)\,.
	$$
\end{lemma}

Concerning Moyal brackets, the situation is slightly more delicate,
but reproducing the standard proof (see e.g. \cite{Tay}, \cite{SVA}) one easily
gets the following result.

\begin{lemma}\label{a.lem comp 2}
	Let $A = \Op(a) \in \OPS^{m, \delta}$ and $B = \Op(b) \in \OPS^{m^\prime, \delta}\,;$ then
	\begin{enumerate}
		\item $A B \in \OPS^{m + m^\prime, \delta}\,.$ Let $a \sharp b$ be its symbol: 
		for any $N_1, N_2 \in \N,$ there exist $K>0$ and $\widetilde{N}_1 > N_1,$ depending only on $N_1, N_2, {\frak c}, p, m, m', \delta,$ such that  
		\begin{equation}\label{a.sonno.4}
		\sup_{\begin{subarray}{c}
			N_1^\prime \leq N_1,  N_2^\prime \leq N_2
			\end{subarray}} C^{}_{N^\prime_1, N^\prime_2} (a \sharp b) \leq K \sup_{\begin{subarray}{c}
			{N}_1^\prime \leq \widetilde{N}_1, {N}_2^\prime \leq {N}_2
			\end{subarray}} C^{}_{{N}^\prime_1, {N}^\prime_2}(a) \sup_{\begin{subarray}{c}
			{N}_1^\prime \leq \widetilde{N}_1, {N}_2^\prime \leq {N}_2
			\end{subarray}} C^{}_{{N}^\prime_1, {N}^\prime_2}(b)\,. 
		\end{equation}
		\item 
		$-i[A, B] \in \OPS^{m + m^\prime -\delta, \delta}\,$ and the
		seminorms of its symbol, denoted by $\{a,
		b\}_{\cal M}$, are controlled as follows: for all
		$N_1$ and $N_2 \in \N$ there exist $K>0$ and $\widetilde{N}_1 > N_1,$ depending only on $N_1, N_2, p,
		m, m', \delta, {\frak c},$ such that 
		\begin{equation}
		\sup_{\begin{subarray}{c}
			N^\prime_1 \leq N_1, N^\prime_2 \leq N_2
			\end{subarray} } C^{}_{N^\prime_1, N^\prime_2}(\{a, b\}_{\cal M}) \leq K \sup_{\begin{subarray}{c}
			N^\prime_1 \leq \widetilde{N}_1 \\ N^\prime_2 \leq N_2 + 1
			\end{subarray}} C^{}_{N_1^\prime, N^\prime_2}(a) \sup_{\begin{subarray}{c}
			N^\prime_1 \leq \widetilde{N}_1 \\ N^\prime_2 \leq N_2 + 1
			\end{subarray}}  C^{ }_{N^\prime_1, N^\prime_2+1}(b)\,.
		\end{equation}
		\item If $a$ is a quadratic polynomial in $\xi$, independent of $x$,
		then $$\{a,
		b\}_{\cal M}=\{a,
		b\}\ .$$  
	\end{enumerate}
\end{lemma}
Finally, we recall that the following result holds:
\begin{lemma}[\bf Egorov Theorem]\label{Egorov astratto}
	Let $\delta > 0$, $\change{\eta <\delta},$ $m \in \R$, $G := {\rm Op}^W(g) \in OPS^{\eta, \delta}$ and $A := {\rm Op}^W(a) \in OPS^{m, \delta}$. Then the following holds. 
	\begin{enumerate}
		\item  The linear operator $H := e^{i  G }  A e^{- i  G} \in OPS^{m, \delta}$ and its symbol $h( x, \xi)$ admits an asymptotic expansion of the form 
		$$
		h = a  + \{ a; g \}_{\cal M} + S^{m - 2(\eta + \delta), \delta}\,.
		$$
		\item \change{If $\eta \leq 0$,} for any $\tau \in [- 1, 1]$, $e^{i \tau G} \in OPS^{0, \delta}$. In particular, if $\sigma$ is its symbol, for all $N_1, N_2 \in \N$ one has that there exist $K_1, K_2 >0$ and $\widetilde{N}_1 > N_1$, depending only on $N_1, N_2, {\frak c}, \delta, {d},$ such that
		\begin{equation}
		\sup_{\begin{subarray}{c}
			N^\prime_1 \leq {N}_1, N^\prime_2 \leq N_2
			\end{subarray}} C^{}_{N^\prime_1, N^\prime_2}(\sigma) \leq K_1 \sup_{\begin{subarray}{c}
			N^\prime_1 \leq \widetilde{N}_1, N^\prime_2 \leq N_2
			\end{subarray}}  e^{K_2 C^{}_{N_1^\prime, N^\prime_2}(g)}\,.
		\end{equation}	
	\end{enumerate}
\end{lemma}
\section{Technical Lemmas}\label{lemmacci}

We first recall Lemma 5.7 of \cite{GioPisa}.

\begin{lemma}[Lemma 5.7 of \cite{GioPisa}] \label{giorgilli tipo kramer}
	Let $s \in \{1\,, \dots\,, d\}$ and let $\{ u_1\,, \dots u_s \}$ be linearly independent vectors in $\R^d\,$ {equipped with the euclidean metric $| \cdot |\,.$} Denote by $\textrm{Vol} \{ u_1\,| \cdots\,| u_s \}$ the $s-$ dimensional volume of the parallelepiped with sides $ u_1\,, \dots\,, u_s\,.$ Let moreover $w \in \Span{ \{ u_1\,, \dots u_s \}}$ be any vector. If there exists positive constants $\alpha\,, N$ such that
	\[
	\begin{gathered}
	|u_j| \leq N \quad \forall j = 1\,, \dots s\,,\\
	| w \cdot u_j | \leq \alpha \quad \forall j = 1\,, \dots s\,,
	\end{gathered}
	\]
	then
	\[
	|w| \leq \frac{s N^{s-1} \alpha }{\textrm{Vol} \{u_1\,| \cdots\,| u_s \}}\,.
	\]
\end{lemma}
We remark that, since all the quantities involved in the statement are
coordinate independent, Lemma \ref{giorgilli metrica g} immediately
follows from it.

\begin{lemma} \label{lemma parallelepipedi}
	 {Let $\{e_1, \dots, e_d\}$ be the vectors of the standard basis in $\R^d.$ There exists a positive constant ${\frak C},$ depending only on
	\begin{equation}\label{c equiv}
	{{\frak c}_2 =  \underset{j = 1, \dots, d}{\max} \norma{e_j}}
	\end{equation}
	and
	\begin{equation}\label{maglia}
	{\frak v} = \int_{\T^d} d \mu_g(x)\equiv\mu_g\left(\T^d\right)\,,
	\end{equation}
	 such that for any $s \in \{1\,, \dots\,, d\}$ and for any set $\{ u_1\,, \dots u_s \}$ of linearly independent vectors in $\Z^d\,$
	\[
	\textrm{Vol}_g \{ u_1\,| \cdots\,| u_s  \} \geq {\frak C}\,.
	\] }
\end{lemma}
\begin{proof}
	We observe that, if $\{ e_1\,, \dots\,, e_d\} $ is the canonical basis of $\Z^d,$ there exists a subset $\{ u^\prime_{s +1}, \dots u^\prime_{d} \} \subset \{ e_{1}, \dots, e_{d}\} $ such that
	$$
	\{ u_{1}, \dots, u_{s}, u^\prime_{s+1}, \dots u^\prime_{d} \}
	$$
	is a set of linearly independent vectors in $\Z^d\,.$ Hence one has that, if $M$ is the linear subspace generated by $\{u_1, \dots, u_s\},$
	\begin{align*}
	\textrm{Vol}_g \{ u_{1}| \dots| u_{s}| u^\prime_{s+1}| \dots |u^\prime_{d} \} & \leq  \norma{u^\prime_{s+1}} \cdots \norma{ u^\prime_{d}}\ \textrm{Vol}_g \{ u_{1}| \dots| u_{s} \}\\
	& \leq ({\frak c}_2)^{d} \textrm{Vol}_g \left(\{ u_{1}| \dots| u_{s} \} \right)\,,
	\end{align*}
	{by the definition of ${\frak c}_2$ as in \eqref{c equiv}.}
	In particular, one has that
	\begin{equation} \label{volume}
	\textrm{Vol}_g \left(\{ u_{1}| \dots| u_{s} \} \right) \geq ({\frak c}_2)^{-d}\	\textrm{Vol}_g \{ u_{1}| \dots| u_{s}| u^\prime_{s+1}| \dots |u^\prime_{d} \} \,.
	\end{equation}
Write	$$
	u_j = \sum_{k = 1}^{d} n_{j, k} e_k\,, \quad u^\prime_j = \sum_{k = 1}^{d} n_{j, k} e_k\,,
	$$
	and if $\forall k = 1, \dots, d$ $\tilde{e}_k$ is the vector of the components of $e_k$ with respect to an orthonormal basis for the inner product $\intg{\cdot}{\cdot},$ then
	\begin{align*}
	\textrm{Vol}_g \left(\{ u_{1}| \dots| u_{s}| u^\prime_{s +1}| \dots| u^\prime_{d} \} \right) & = \textrm{Vol}_g \left(\left \lbrace \sum_{k =1}^{d} n_{1, k} e_k | \dots | \sum_{k =1}^{d} n_{d, k} e_k \right \rbrace \right)\\
	&= \textrm{Vol} \left(\left \lbrace \sum_{k =1}^{d} n_{1, k} \tilde{e}_k | \dots | \sum_{k =1}^{d} n_{d, k} \tilde{e}_k  \right \rbrace \right)\\
	& \geq \textrm{Vol}\left( \tilde{e}_{{1}} | \cdots | \tilde{e}_{{d}}\right)\\
	& = \textrm{Vol}_g \left( {e}_{{1}} | \cdots | {e}_{{d}}\right) = {\frak v}\,.
	\end{align*}
	Thus \eqref{volume} implies that
	\[
	\textrm{Vol} \{ u_1 \ | \cdots\ |\ u_s \} \geq ({\frak c}_2)^{-d} {\frak v} =: {\frak C}\,.
	\]
\end{proof}

\begin{remark} \label{rmk triangolare}
By studying the function $(1+x^2)^{a/2}$ it is easy to see that
there exists a constant $K$ s.t. $\forall \xi,\eta\in\R^d$ one has
\begin{equation}
  \label{tri.1}
\langle \xi+\eta\rangle^a\leq K(\langle \xi\rangle^a+\langle \eta\rangle^a)\ . 
\end{equation}
Furthermore, since, for any $C>0$
$$
\sup _{y>C}\frac{\langle y\rangle}{y}<\infty\ ,
$$
one also has $\exists K'=K'(a,C)$ s.t.
\begin{equation}
  \label{tri.2}
\langle \xi+\eta\rangle^a\leq K'(\langle \xi\rangle^a+\norm{\eta}^a)\ ,\quad \forall
\eta\ : \norm\eta\geq C \ .
\end{equation}
\end{remark}

\begin{remark}
  \label{B.11.1} If $\norma{\xi-\eta}\leq F\langle \xi\rangle^a$, with $a<1$,
  one has
  $$
\langle \xi\rangle\leq K(1+F)\langle \eta\rangle\ .
  $$

\end{remark}

A further useful lemma is the following one

\begin{lemma}
  \label{xallaa}
  Let $N\geq 1$, $a<1$, $K\geq 2^{-a}$ be positive real numbers, Then
  \begin{equation}
    \label{xalla.1}
x-Kx^a\leq N\ \Longrightarrow x\leq(2K)^{\frac{1}{1-a}}N\ .
  \end{equation}
\end{lemma}
\proof If $Kx^a\leq \frac{x}{2}$, which is equivalent to
\begin{equation}
  \label{xal.1}
x\geq
(2K)^{\frac{1}{1-a}}\ ,
\end{equation}then the assumed inequalities implies
$$
\frac{1}{2}x\leq x-Kx^a\leq N\ \Longrightarrow\ x<2N\ ,
$$ but, by assumption, the r.h.s is smaller than $(2K)^{\frac{1}{1-a}}
$, and therefore the thesis holds in this case. On the contrary, the
converse of \eqref{xal.1}, implies
$$
x<(2K)^{\frac{1}{1-a}}\leq (2K)^{\frac{1}{1-a}}N\ ,
$$
which again implies the thesis. \qed

\begin{lemma}
	\label{xixik}
	Let $1>a>\epsilon>0$ and $1>\delta>0$ be parameters. Let $\sigmav$,
	$\eta$, $k$, $\ell$ be vectors. Assume that there exist constants
	$C,F,D,D_0$ s.t.
	\begin{align}
	\label{xixi.1}
	\left|\scala{\sigmav+k}h\right|\leq C\langle\sigmav+k\rangle^\delta {|h|^{-\tau}}\ ,
	\\
	\norm{k}\leq D\langle\sigmav+k\rangle^{\epsilon}\ ,\quad \norm h\leq
	D_0\langle\sigmav+k\rangle^\epsilon
	\\
	\label{xixi.2}
	\norm{\eta-\sigmav}\leq F\langle\eta\rangle^a\ ,\quad \norm{\ell}\leq
	D\langle\eta+\ell\rangle^\epsilon\ ;
	\end{align}
	then there exists $K'$ and $D'$ (which depends on the above constants), s.t.
	\begin{align}
	\label{est.v.1}
	\langle\sigmav+k\rangle\leq 
	D' \langle\eta+\ell\rangle\ ,
	\\
	\label{xixi.5}
	\left|\scala{\eta+\ell}h\right|\leq
	K'\langle\eta+\ell\rangle^{\max\{\delta,{a+\epsilon(\tau + 1)}  \}} {|h|^{-\tau}} \ .
	\end{align}
\end{lemma}
\proof Start by writing
\begin{align}
\label{x.w1}
\sigmav+k&=\eta+\ell+v
\\
v&:=k-\ell+\sigmav-\eta\ ;  
\end{align}
then we estimate $v$ (with $\eta+\ell$). One has
\begin{align*}
\norm{v}&\leq
D\langle\sigmav+k\rangle^\epsilon+D\langle\eta+\ell\rangle^\epsilon+F\langle\eta
\rangle^a\\
&=D\langle\eta+\ell+v
\rangle^\epsilon+D\langle\eta+\ell\rangle^\epsilon+F\langle\eta+\ell-\ell
\rangle^a
\\
&\leq DK\left(\langle\eta+\ell\rangle^\epsilon +\langle
v\rangle^\epsilon \right) +D\langle\eta+\ell\rangle^\epsilon+FK
\left(\langle\eta+\ell
\rangle^a+\langle\ell\rangle^a\right)
\\ &\leq
D(K+1)\langle\eta+\ell\rangle^\epsilon+FK\langle\eta+\ell\rangle^a+FK(1+D)
\langle\eta+\ell \rangle^{a\epsilon}
+DK\langle v\rangle^\epsilon  \ .
\end{align*}
Using $a>\epsilon$ and $a>a\epsilon$, (and exploiting $\langle
x\rangle\leq 1+x$, which holds for all positive $x$) we get
\begin{align*}
\langle v\rangle\leq \left(D(K+1)+FK+FK(1+D)+1\right)\langle\eta+\ell
\rangle^a
\\
+DK\langle v\rangle^\epsilon\ .
\end{align*}
Applying Lemma \ref{xallaa} with $N$ equal to the first line, we get
that there exists a constant $K''$ (explicitly computable), s.t.
\begin{equation}
\label{est.v}
\langle v\rangle\leq K''\langle\eta+\ell\rangle^a\ .
\end{equation}
Exploiting this and using again \eqref{x.w1}, we immediately get
\eqref{est.v.1}.  
We are now ready for the final estimate:
\begin{align*}
\left|\scala{\eta+\ell}h\right|& \leq
\left|\scala{\sigmav+k}h\right| +\left|\scala{v} h\right| {\|h\|^{\tau}} {\|h\|^{-\tau}}
\\
& \leq
C\langle\sigmav+k
\rangle^\delta{\|h\|^{-\tau}} +K''\langle\eta+
\ell\rangle^aD_0\langle\sigmav+k\rangle^\epsilon {D_0^\tau \langle \sigmav + k \rangle ^{\ep \tau}} {\|h\|^{-\tau}}\\
& \leq C{(D^\prime)^\delta}\langle\eta+\ell\rangle^\delta {\|h\|^{-\tau}} +K''\langle\eta
+\ell\rangle^a D_0^{\tau + 1} {(D^\prime)^{\ep(\tau + 1)}} \langle\eta+\ell 
\rangle^{\epsilon(\tau + 1)} {\|h\|^{-\tau}}\,,\\
\end{align*}
from which the thesis immediately follows. \qed

\section{Properties of eigenfunctions}\label{aquasi}

\begin{lemma} \label{prop struttura autof}
  Consider an operator $H_0 + \cR$, with $\cR\in\ops{-\tN\delta}$;
  assume that
  \begin{enumerate}
\item $\exists C$ and $d$ s.t. the spectrum of $H_0$ satisfies a Weyl's law of the form
\begin{equation}\label{ancora weyl}
		\sharp \{\lambda^{(0)} \in
                \sigma(H_0)\ |\ \lambda^{(0)} \leq r\} \leq {
                  C}r^{\frac{d}{2}}\,. 
		\end{equation}
\item There exist $a>0$ and ${C}_1$
                  such that any normalized eigenfunction $\psi$
                  relative to an eigenvalue ${\lambda}^{(0)}$ of $H_0$
                 fulfills
		\begin{equation}\label{r piccolo}
		\|\psi\|_{H^{-\tN}} \leq {C_1} |{\lambda^{(0)}}|^{-a\tN}\,.
		\end{equation}
	\end{enumerate}
  Then there exists $\Lambda,C_1'>0$ which depend on ${C}, C_1, {d}, \| R\|_{{\cal B}({H}^{-N}, {H}^{0})}$ only, with the following
  properties: any normalized eigenfunction 
  $\phi$ of $H_0 + R$ which corresponds to an eigenvalue
  $\lambda>\Lambda$ fulfills
	\begin{equation}\label{piccola}
		{\|\phi\|_{H^{-\tN}} \leq {C}_1^\prime |{\lambda}|^{\frac{ d}{2} -a\tN}\,.}
		\end{equation}
	\end{lemma}
\begin{proof}
  First remark that, by the Calderon Vaillancourt theorem, one has
  \begin{equation}
    \label{rpsi}
{\norma{R\psi}_{L^2}}\leq \| R\|_{{\cal B}({H}^{-\tN}, {H}^{0})} \norma{\psi}_{H^{-\tN}}\leq
\frac{\| R\|_{{\cal B}({H}^{-\tN}, {H}^{0})} C_1}{\left|\lambda^{(0)}\right|^{a\tN}}\ .
  \end{equation}
Fix $\tc_1<\lambda/2$ and decompose   $$
\phi=\phi_0+\phi_1
  $$ with $$\phi_0\in{\cal Q} = \Span \left \lbrace \psi \ \left|\ H_0
\psi = \lambda_\psi \psi\,, \quad |\lambda_\psi - \lambda| \leq {\tt
  c}_1\right.\right \rbrace\,;$$ and $\phi_1\in
\cQ^{\perp}$.  We analyze the eigenvalue equation
	$$
	\left(H_0 + R\right)\phi = \lambda \phi\,.
	$$ by using the method of Lyapunov Schmidt
        decomposition. Denote by $\Pi^\perp$ the orthogonal projector
        on $\cQ^{\perp}$ and by $\Pi$ the orthogonal projector on
        $\cQ$.  Inserting the decomposition of $\phi$
        in the eigenvalue equation, applying $\Pi^{\perp}$  and taking
        into account that the projector commutes with $H_0$, we get
        (reorganizing the terms)
        $$
\left[\left(\Pi^{\perp}H_0\Pi^{\perp}-\lambda\right)+\Pi^{\perp}R
  \right]\phi_1 =-\Pi^{\perp} R\phi_0\ .
        $$
By definition of $\cQ^{\perp}$, the operator in square brackets is invertible and the
norm of its inverse is bounded by 2, provided ${\tt c}_1 \geq 2\|
R\|_{{\cal B}\left({H^{-\tN}, H^0}\right)} .$ It follows that
$$
{\norma{\phi_1}_{L^2}}\leq 2{\norma{R\phi_0}_{L^2}}\ .
$$
To estimate ${\norma{R\phi_0}_{L^2}}$ we decompose $\phi_0$ in eigenfunctions
of $H_0$ and use assumption \eqref{r piccolo}. First remark that by
construction $\phi_0$ has components only on eigenfunctions
corresponding to eigenvalues between $\lambda-\tc_1>\lambda/2$ and
$\lambda+\tc_1<2\lambda$. So it has at most $J\leq 2^{d/2}C\lambda^{\td/2}$
components:
$$
	\phi_0 = \sum_{j = 1}^{J} \alpha_j \psi_j\, .
	$$ It follows that the $H^{-\tN}$ norm of $\phi_0$ is bounded
        by $2^{\tN a}J/\lambda^{a\tN}$.  Concerning $\phi_1$, we show that its
        $L^2$ norm, which bounds all the negative Sobolev norms, is
        small. One has
	$$
	{\| R \phi_0\|_{L^2}} \leq \sum_{j = 1}^{J} |\alpha_j| {\| R \psi_j\|_{L^2}}
        \leq \left(\sum_{j = 1}^{J} |\alpha_j|^2\right)^{\frac{1}{2}}
        \left(\sum_{j = 1}^{J}{\|R \psi_j\|_{L^2}^2}\right)^{\frac{1}{2}} \leq
        J^{1/2}\frac{  \norma{R}_{\cB(H^{-\tN},H^0)}}{(\lambda/2)^{a\tN}}\,,
	$$
        where we used that the norm of $\phi_0$ is smaller than the
        norm of $\phi$ and therefore is smaller than 1. From this the
        thesis follows. 
\end{proof}

\begin{lemma}
  \label{negativo}
  Let $M$ be a module, and let $u$ be a function of the form
  $$
u(x)=\sum_{\zeta\in M}\hat u_{\zeta}e^{\im \zeta\cdot x}\ ,
$$
be such that 
\begin{equation}
  \label{nega.2}
\norma{u}_{H^{-\tN}}\leq K
\end{equation}
Let $\beta\in  {\Mc}$ and consider $\tilde \beta$ defined as in \eqref{betainz}, then one has
\begin{equation}
  \label{negami}
\norma{e^{\im\tilde \beta\cdot x}u}_{H^{-2\tN}}\leq \frac{K}{\langle
  (\tilde \beta+\kappa)_{M^{\perp}}\rangle^\tN}\ .
\end{equation}
 \end{lemma}
\proof One has
\begin{equation}
  \label{negamipure}
\norma{e^{\im\tilde \beta\cdot x}u}^2_{H^{-2N}}=\sum_{\zeta\in
  M}\langle \tilde
\beta+\kappa+\zeta\rangle^{-2\tN}\left|\hat u_{\zeta}\right|^2\ .
\end{equation}
We analyse, using \eqref{intern} and \eqref{intern.2}, the term 
\begin{align*}
\langle \tilde
\beta+\kappa+\zeta\rangle^2=1+(\zeta+\tilde\beta+\kappa)_M^2+(\zeta+\tilde\beta+\kappa)_{M^\perp}^2
\\
=1+(\zeta+\kappa')^2+(\tilde
\beta+\kappa)_{M^{\perp}}^2=\frac12+(\zeta+\kappa')^2 +\frac12+(\tilde
\beta+\kappa)_{M^{\perp}}^2
\\
\geq 2\sqrt{\frac12+(\zeta+\kappa')^2 }\sqrt  {\frac12+(\tilde
\beta+\kappa)_{M^{\perp}}^2}\geq \langle \zeta+\kappa' \rangle
\langle (\tilde
\beta+\kappa)_{M^{\perp}}   \rangle \ .
\end{align*}
Inserting in \eqref{negamipure} one immediately gets the thesis.\qed

\def\cprime{$'$}

\end{document}